\documentclass[10pt,english]{IEEEtran}
\usepackage[T1]{fontenc}
\usepackage[latin9]{inputenc}
\usepackage{array}
\usepackage{verbatim}
\usepackage{amsmath}
\usepackage{amssymb}
\usepackage{graphicx}
\usepackage{esint}

\makeatletter

\providecommand{\tabularnewline}{\\}

\makeatother

\usepackage{amsthm}\usepackage{dsfont}\usepackage{array}\usepackage{mathrsfs}\usepackage{cite}\usepackage{comment}\usepackage{mathrsfs}\usepackage{hyperref}

\makeatother

\usepackage{babel}

\makeatother

\usepackage{babel}
\begin{document}
\bibliographystyle{IEEEtran}

\title{Channel Capacity under \\Sub-Nyquist Nonuniform Sampling}

\author{Yuxin Chen, Andrea J. Goldsmith, and Yonina C. Eldar%
\thanks{Y. Chen is with the Department of Electrical Engineering, Stanford
University, Stanford, CA 94305, USA (e-mail: yxchen@stanford.edu). 

A. J. Goldsmith is with the Department of Electrical Engineering,
Stanford University, Stanford, CA 94305, USA (e-mail: andrea@ee.stanford.edu). 

Y. C. Eldar is with the Department of Electrical Engineering, Technion,
Israel Institute of Technology Haifa, Israel 32000 (e-mail: yonina@ee.technion.ac.il). 

This work was supported in part by the NSF Center for Science of Information,
the Interconnect Focus Center of the Semiconductor Research Corporation,
and BSF Transformative Science Grant 2010505. This paper was presented
in part at the 2012 IEEE International Symposium on Information Theory
\cite{ChenNonuniform2012}. 

Copyright (c) 2012
IEEE. Personal use of this material is permitted.  However, permission
to use this material for any other purposes must be obtained from the
IEEE by sending a request to pubs-permissions@ieee.org.%
}}
\maketitle
\begin{abstract}
This paper investigates the effect of sub-Nyquist sampling upon the
capacity of an analog channel. The channel is assumed to be a linear
time-invariant Gaussian channel, where perfect channel knowledge is
available at both the transmitter and the receiver. We consider a
general class of right-invertible time-preserving sampling methods
which include irregular nonuniform sampling, and characterize in closed
form the channel capacity achievable by this class of sampling methods,
under a sampling rate and power constraint. Our results indicate that
the optimal sampling structures extract out the set of frequencies
that exhibits the highest signal-to-noise ratio among all spectral
sets of measure equal to the sampling rate. This can be attained through
filterbank sampling with uniform sampling at each branch with possibly
different rates, or through a single branch of modulation and filtering
followed by uniform sampling. These results reveal that for a large
class of channels, employing irregular nonuniform sampling sets, while
typically complicated to realize, does not provide capacity gain over
uniform sampling sets with appropriate preprocessing. Our findings
demonstrate that aliasing or scrambling of spectral components does
not provide capacity gain, which is in contrast to the benefits obtained
from random mixing in spectrum-blind compressive sampling schemes.\end{abstract}
\begin{IEEEkeywords}
nonuniform sampling, irregular sampling, sampled analog channels,
sub-Nyquist sampling, channel capacity, Beurling density, time-preserving
sampling systems
\end{IEEEkeywords}
\theoremstyle{plain}\newtheorem{lem}{\textbf{Lemma}}\newtheorem{theorem}{\textbf{Theorem}}\newtheorem{corollary}{\textbf{Corollary}}\newtheorem{prop}{\textbf{Proposition}}\newtheorem{fct}{\textbf{Fact}}\newtheorem{remark}{\textbf{Remark}}

\theoremstyle{definition}\newtheorem{definition}{\textbf{Definition}}\newtheorem{example}{\textbf{Example}}

\section{Introduction}

The capacity of analog Gaussian channels and their capacity-achieving
transmission strategies were pioneered by Shannon \cite{Sha48}, which
has provided fundamental insights for modern communication system
design. Shannon's work focused on capacity of analog channels sampled
at or above twice the channel bandwidth. However, these results do
not explicitly account for sub-Nyquist sampling rate constraints that
may be imposed by hardware limitations. This motivates exploration
of the effects of sub-Nyquist sampling upon the capacity of an analog
Gaussian channel, and the fundamental capacity limits that result
when considering general sampling methods that include irregular nonuniform
sampling.

\subsection{Related Work and Motivation}

Shannon introduced and derived the information theoretic metric of
channel capacity for time-invariant analog waveform channels\cite{Sha48},
which established the optimality of water-filling power allocation
based on signal-to-noise ratio (SNR) over the spectral domain \cite{Gallager68,HirtMassey1988}.
A key idea in determining the analog channel capacity is to convert
the continuous-time channel into a set of parallel discrete-time channels
based on the Shannon-Nyquist sampling theorem \cite{Bello1963}. This
paradigm was employed, for example, by Medard et. al. to bound the
maximum mutual information in time-varying channels \cite{Med2000,MedGal2002},
and was used by Forney\emph{ }et. al. to investigate coding and modulation
for Gaussian channels \cite{ForUng1998}. Most of these results focus
on the analog channel capacity commensurate with uniform sampling
at or above the Nyquist rate associated with the channel bandwidth.
There is another line of work that characterizes the effects upon
information rates of oversampling with quantization \cite{Gil1993,Sha1994}.
In practice, however, hardware and power limitations may preclude
sampling at the Nyquist rate for a wideband communication system. 

More general irregular sampling methods beyond pointwise uniform sampling
have been extensively studied in the sampling literature, e.g. \cite{AldGro2001,Yeung2001,EldOpp2000}.
One example is sampling on non-periodic quasi-crystal sets, which
has been shown to be stable for bandlimited signals \cite{Matei2008quasicrystals,lev2012riesz}.
These sampling approaches are of interest in some realistic situations
where signals are only sampled at a nonuniformly spaced sampling set
due to constraints imposed by data acquisition devices. Many sophisticated
reconstruction algorithms have been developed for the class of bandlimited
signals or, more generally, the class of shift-invariant signals \cite{Grochenig1992,FeichtingerGrochengig1994,AldGro2001}.
For all these nonuniform sampling methods, the Nyquist sampling rate
is necessary for perfect recovery of bandlimited signals \cite{Beurling_1,Jerri1977,AldGro2001}.

However, for signals with certain structure, the Nyquist sampling
rate may exceed that required for perfect signal reconstruction from
the samples \cite{ButSte1992,MisEld2011}. For example, consider multiband
signals, whose spectral contents reside within several subbands over
a wide spectrum. If the spectral support is known, then the necessary
sampling rate for the multiband signals is their spectral occupancy,
termed the \textit{Landau rate} \cite{Landau1967}. Such signals admit
perfect recovery when sampled at rates approaching the Landau rate,
provided that the sampling sets are appropriately chosen (e.g. \cite{HerleyWong1999,VenkataramaniBresler2001}).
One type of sampling mechanism that can reconstruct multiband signals
sampled at the Landau rate is a filter bank followed by sampling,
studied in \cite{LinVai1998,UnsZer1998,MisEld2009}. Inspired by recent
``compressive sensing'' \cite{CandRomTao06,Don2006,eldar2012compressed}
ideas, spectrum-blind sub-Nyquist sampling for multiband signals with
random modulation has been developed \cite{MisEld2010Theory2Practice}
as well.

Although sub-Nyquist nonuniform sampling methods have been extensively
explored in the sampling literature, they are typically investigated
either under a noiseless setting, or based on statistical reconstruction
measures (e.g. mean squared error (MSE)) instead of information theoretic
measures. Gastpar\emph{ }et al \cite{GastparBresler2000} studied
the necessary sampling density for nonuniform sampling. Recent work
by Wu and Verdu \cite{wu2012optimal} investigated the tradeoff between
the number of samples and the reconstruction fidelity through information
theoretic measures. However, these works did not explicitly consider
the capacity metric for an analog channel. The most relevant capacity
result to our work was by Berger et al\emph{ }\cite{BergerThesis},
who related MSE-based optimal sampling with capacity for several special
types of channels. But they did not derive the sub-Nyquist sampled
channel capacity for more general channels, nor did they consider
nonuniformly spaced sampling. Our recent work \cite{ChenGolEld2010}
established a new framework that characterizes sampled capacity for
a broad class of sampling methods, including filter and modulation
bank sampling \cite{Papoulis1977,MisEld2010Theory2Practice,EldMic2009}.
For these sampling methods, we determined optimal sampling structures
based on capacity as a metric, illuminated intriguing connections
between MIMO channel capacity and capacity of undersampled channels,
as well as a new connection between capacity and MSE. However, this
prior work did not investigate analog channel capacity using more
general nonuniform sampling under a sub-Nyquist sampling rate constraint.

One interesting fact discovered in \cite{ChenGolEld2010} is the non-monotonicity
of capacity with sampling rate under filter- and modulation-bank sampling,
assuming an equal sampling rate per branch for a given number of branches.
This indicates that more sophisticated sampling techniques, adaptive
to the channel response and the sampling rate, are needed to maximize
capacity under sub-Nyquist rate constraints, including both uniform
and nonuniform sampling. However, none of the aforementioned work
has investigated the question as to which sampling method can best
exploit channel structure, thereby maximizing sampled capacity under
a given sampling rate constraint. Although several classes of sampling
methods were shown in \cite{ChenGolEld2010} to have closed-form capacity
solutions, the capacity limits might not even exist for general sampling
methods. This raises the question as to whether there exists a capacity
upper bound over a general class of sub-Nyquist sampling systems beyond
the classes we discussed in \cite{ChenGolEld2010} and, if so, when
the bound is achievable. That is the question we investigate herein.

\subsection{Contributions and Organization}

Our main contribution is to derive the capacity of sub-Nyquist sampled
analog channels for a general class of right-invertible time-preserving
nonuniform sampling methods, under a sub-Nyquist sampling rate constraint.
The channel is assumed to be a linear time-invariant (LTI) Gaussian
channel, where perfect channel knowledge is available at both the
transmitter and the receiver. The class of sampling systems we consider
subsumes sampling structures employing irregular nonuniform sampling
grids. 

We first develop in Theorem \ref{thm:GeneralSampledCapacity} an upper
bound on the sampled channel capacity, which corresponds to the capacity
of a channel whose spectral occupancy is no larger than the sampling
rate $f_{s}$. As a key step in the analysis framework for Theorem
\ref{thm:GeneralSampledCapacity}, we characterize in closed form
the sampled channel capacity for any specific periodic sampling system,
formally defined in Definition \ref{definition-PeriodicSampling}
(Lemma \ref{thm:PeriodicSampledCapacity}). We demonstrate that this
fundamental capacity limit can be achieved by filterbank sampling
with \emph{varied} sampling rates at different branches, or by a single
branch of modulation and filtering followed by a uniform sampling
set (Theorems \ref{thm:OptimalSamplingGeneralSampledCapacity}-\ref{theorem:AchievabilityModulation}).
In particular, the optimal sampler extracts out a spectral set of
size $f_{s}$ with the highest SNR, and suppresses all signal and
noise components outside this spectral set. 

Our results indicate that irregular nonuniform sampling sets, while
typically complicated to realize in hardware, do not increase channel
capacity relative to analog preprocessing with regular uniform sampling
sets. We also show that when optimal filterbank or modulation sampling
is employed, a mild perturbation of the optimal sampling grid does
not change the capacity. Our findings demonstrate that aliasing or
scrambling of spectral contents does not provide capacity gain. This
is in contrast to the benefits obtained from random mixing of frequency
components in many sub-Nyquist sampling schemes with unknown signal
support (e.g. \cite{MisEld2010Theory2Practice}). 

The main innovation of this paper compared to our previous sub-sampled
channel capacity results in \cite{ChenGolEld2010} is as follows.
\begin{itemize}
\item While \cite{ChenGolEld2010} characterizes the capacity under two
types of sampling mechanisms that are widely used in practice (filter-bank
sampling and modulation-bank sampling), the focus of this paper is
instead to develop capacity results over a much more general class
of sampling methods. 
\item While all results of \cite{ChenGolEld2010} hold only under uniform
sampling, our analysis herein accommodates irregular nonuniform sampling.
Our results in turn corroborate the optimality of uniform sampling
in achieving sampled capacity, assuming that the analog channel output
is appropriately pre-processed.
\end{itemize}
The remainder of the paper is organized as follows. In Section \ref{sec:Sampled-Channel-Capacity},
we introduce our system model of sampled analog channels, and provide
formal definitions of time-preserving systems, sampling rates, and
sampled channel capacity. We then develop, in Section \ref{sub:CapacityUpperBound},
an upper bound on the sampled channel capacity ranging over all right-invertible
time-preserving sampling methods, along with an approximate analysis
highlighting insights into the result. The achievability of this upper
bound is derived in Section \ref{sub:Achievability}. The proof of
Theorem \ref{thm:GeneralSampledCapacity} is provided in Appendix
\ref{sec:Proof-Architecture-of-Theorem-General-Capacity}. The implications
of our main results are summarized in Section \ref{sec:Discussion}. 

Before continuing, we introduce some notation that will be used throughout.
We use $\mu\left(\cdot\right)$ to represent the Lebesgue measure,
and denote by $\mathcal{F}$ and $\mathcal{F}^{-1}$ the Fourier and
inverse Fourier transform, respectively. We let $[x]^{+}\overset{\Delta}{=}\max\left(x,0\right)$,
and use $\mathrm{card}\left(A\right)$ to denote the cardinality of
a set $A$. These and other notation in the paper are summarized in
Table \ref{tab:Summary-of-Notation-Nonuniform}.

\begin{table}
\caption{\label{tab:Summary-of-Notation-Nonuniform}Summary of Notations and
Parameters}

\centering{}%
\begin{tabular}{>{\raggedright}p{2cm}>{\raggedright}p{2.3in}}
$\mu(\cdot)$ & Lebesgue measure\tabularnewline
$\Lambda$ & sampling set $\left\{ t_{n}:n\in Z\right\} $\tabularnewline
$D^{+}(\Lambda),$ $D^{-}(\Lambda)$, $D(\Lambda)$ & upper, lower and uniform Beurling densities of $\Lambda$\tabularnewline
$\mathcal{L}_{2}\left(\Omega\right)$ & set of measurable functions $f$ supported on the set $\Omega$ such
that $\int\left|f\right|^{2}\mathrm{d}\mu<\infty$ \tabularnewline
$\mathbb{S}_{+}$ & set of positive semidefinite matrices\tabularnewline
$h(t)$,$H(f)$  & impulse response and frequency response of the LTI analog channel\tabularnewline
$s_{i}(t)$, $S_{i}(f)$  & impulse response and frequency response of the $i$th (post-modulation)
filter\tabularnewline
$p(t)$, $P(f)$  & impulse response and frequency response of the pre-modulation filter\tabularnewline
$\mathcal{S}_{\eta}(f),s_{\eta}(t)$ & power spectral density of the noise $\eta(t)$ and $s_{\eta}\left(t\right):=\mathcal{F}^{-1}\left(\sqrt{\mathcal{S}_{\eta}\left(f\right)}\right)$\tabularnewline
$f_{s}$, $T_{s}$ & aggregate sampling rate and the corresponding sampling interval ($T_{s}=1/f_{s}$)\tabularnewline
$q(t,\tau)$ & impulse response of the sampling system, i.e. the output seen at time
$t$ due to an impulse in the input at time $\tau$.\tabularnewline
$T_{q},f_{q}$ & period of the modulating sequence $q(t)$ such that $T_{q}=1/f_{q}$ \tabularnewline
$\mathcal{F}$, $\mathcal{F}^{-1}$ & Fourier transform and inverse Fourier transform\tabularnewline
\end{tabular}
\end{table}

\section{Sampled Channel Capacity \label{sec:Sampled-Channel-Capacity}}

\subsection{System Model}

We consider an analog waveform channel, which is modeled as an LTI
filter with impulse response $h(t)$ and frequency response $H(f)=\int_{-\infty}^{\infty}h(t)\exp(-j2\pi ft)\text{d}t$.
With $x(t)$ denoting the transmitted signal, the analog channel output
is given by
\begin{equation}
r(t)=h(t)*x(t)+\eta(t),\label{eq:ChannelModel}
\end{equation}
where the noise process $\eta(t)$ is assumed to be an additive stationary
zero-mean Gaussian process with power spectral density $\mathcal{S}_{\eta}\left(f\right)$.
We also define $s_{\eta}\left(t\right):=\mathcal{F}^{-1}\left(\sqrt{\mathcal{S}_{\eta}\left(f\right)}\right)$.
Unless otherwise specified, we assume throughout that \textit{perfect
channel state information} (i.e. the knowledge of both $H(f)$ and
$\mathcal{S}_{\eta}(f)$) is available at both the transmitter and
the receiver. 

The analog channel output $r(t)$ is passed through $M$ ($1\leq M\leq\infty$)
branches of linear preprocessing systems, each followed by a pointwise
sampler, as illustrated in Fig. \ref{fig:ProblemFormulation}. The
preprocessed output $y_{i}(t)$ at the $i$th branch is obtained by
applying a linear bounded operator $\mathcal{T}_{i}$ to the channel
output $r(t)$:
\begin{align}
y_{i}(t) & =\mathcal{T}_{i}\left(r(t)\right)=\int q_{i}\left(t,\tau\right)r\left(\tau\right)\mathrm{d}\tau,
\end{align}
where $q_{i}(t,\tau)$ denotes the impulse response of the time-varying
system represented by $\mathcal{T}_{i}$, i.e. the output seen at
time $t$ due to an impulse in the input at time $\tau$. Note that
the linear operator $\mathcal{T}_{i}$ can be time-varying, which
subsumes filtering and modulation as special cases. For example, a
modulation system $\mathcal{T}_{i}\left(x(t)\right)=p(t)x(t)$ for
some given modulation sequence $p(t)$ has an impulse response $q(t,\tau)=p(\tau)\delta\left(t-\tau\right)$.
A cascade combination of two systems $\mathcal{T}_{1}$ and $\mathcal{T}_{2}$
has an impulse response $q(t,\tau)=\int_{-\infty}^{\infty}q_{2}(t,\tau_{1})q_{1}(\tau_{1},\tau)\mathrm{d}\tau_{1}$,
with $q_{1}\left(\cdot,\cdot\right)$ and $q_{2}(\cdot,\cdot)$ denoting
respectively the impulse responses of $\mathcal{T}_{1}$ and $\mathcal{T}_{2}$
\cite{Molisch2011}. When an operator $\mathcal{T}$ is LTI, we use
$q(\tau):=q(t,t-\tau)$ as shorthand to represent its impulse response.

The pointwise sampler following the preprocessor can be uniform or
irregular \cite{AldGro2001}. Specifically, the preprocessed output
$y_{i}(t)$ (at the $i$th branch) is sampled at times $t_{i,n}\left(n\in\mathbb{Z}\right)$,
yielding a sample sequence $ $$y_{i}[n]=y_{i}\left(t_{i,n}\right)$.
Here, we define the \emph{sampling set} $\Lambda_{i}$ at the $i$th
branch as 
\begin{equation}
\Lambda_{i}:=\left\{ t_{i,n}\mid n\in\mathbb{Z}\right\} .
\end{equation}
In particular, if $t_{i,n}=nT_{i,s}$, then the sampling set at the
$i$th branch is said to be uniform with period $T_{i,s}$. 

\begin{figure}[htbp]
\begin{centering}
\textsf{\includegraphics[scale=0.35]{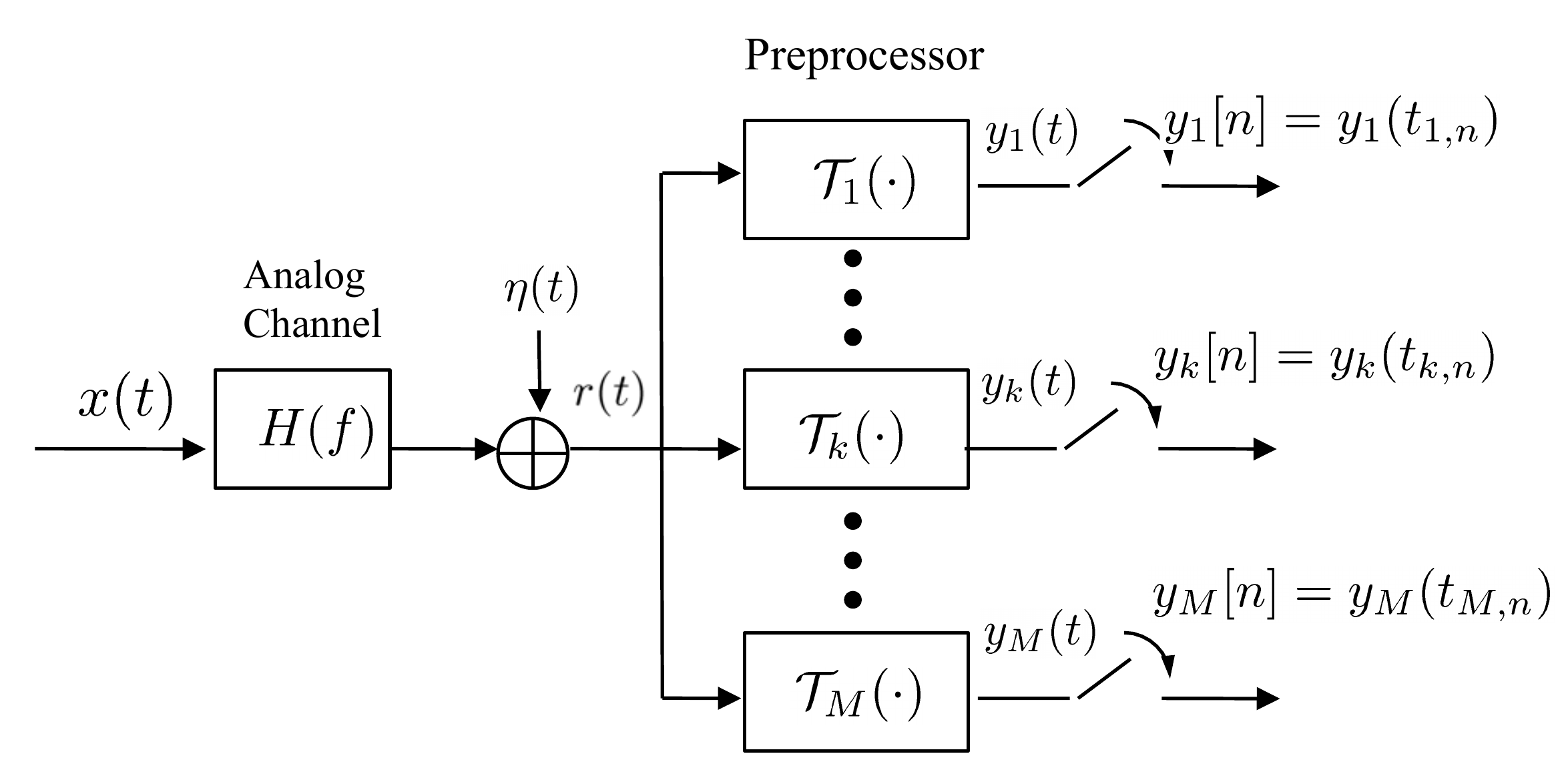}} 
\par\end{centering}

\caption{\label{fig:ProblemFormulation} The input $x(t)$ is passed through
the analog channel and contaminated by noise $\eta(t)$. The analog
channel output $r(t)$ is then passed through $M\text{ }(1\leq M\leq\infty)$
linear preprocessing system $\left\{ \mathcal{T}_{i}\mid1\leq i\leq M\right\} $.
At the $i$th branch, the preprocessed output $y_{i}(t)$ is sampled
on the sampling set $\Lambda_{i}=\left\{ t_{i,n}\mid n\in\mathbb{Z}\right\} $.}
\end{figure}

\subsection{Sampling Rate Definition}

Our metric of interest is the sampled channel capacity under a sampling
rate constraint. We first formally define sampling rate for general
nonuniform sampling mechanisms.

In general, the sampling set $\Lambda=\left\{ t_{n}\mid n\in\mathbb{Z}\right\} $
may be irregular and hence aperiodic, which calls for a generalized
definition of sampling rate. One notion commonly used in sampling
theory is the Beurling density introduced by Beurling \cite{Beurling_1}
and Landau \cite{Landau1967}, as defined below \cite{AldGro2001}.

\begin{definition}[{\bf Beurling Density}]For a sampling set $\Lambda=\left\{ t_{k}\mid k\in\mathbb{Z}\right\} $,
the upper and lower Beurling density are given respectively as
\begin{align*}
D^{+}\left(\Lambda\right) & =\lim_{r\rightarrow\infty}\sup_{z\in\mathbb{R}}\frac{\text{card}\left(\Lambda\cap\left[z,z+r\right]\right)}{r},\\
D^{-}\left(\Lambda\right) & =\lim_{r\rightarrow\infty}\inf_{z\in\mathbb{R}}\frac{\text{card}\left(\Lambda\cap\left[z,z+r\right]\right)}{r}.
\end{align*}
When $D^{+}\left(\Lambda\right)=D^{-}\left(\Lambda\right)$, the sampling
set $\Lambda$ is said to be of uniform Beurling density $D\left(\Lambda\right):=D^{-}\left(\Lambda\right)$.\end{definition}
When the sampling set is uniform with period $T_{s}$, the Beurling
density is $D(\Lambda)=1/T_{s}$, which coincides with the conventional
definition of sampling rate. The notion of Beurling density allows
the Shannon-Nyquist sampling theorem to be extended to nonuniform
sampling. Moreover, we will use Beurling density to define sampling
rate for a large class of sampling mechanisms with preprocessing. 

Under a nonuniform sampling set $\Lambda$, the set of exponential
functions $\left\{ \exp\left(j2\pi t_{n}f\right)\mid n\in\mathbb{Z}\right\} $
forms a \emph{non-harmonic} Fourier series \cite{Yeung2001}. Whether
the class of original signals are recoverable from the nonuniform
sampled sequence is determined by the completeness of the associated
non-harmonic set. In particular, when $\Lambda$ is uniform, the set
$\left\{ \exp\left(j2\pi t_{n}f\right)\mid n\in\mathbb{Z},t_{n}=n/f_{s}\right\} $
with $D(\Lambda)=f_{s}$ forms a Riesz basis \cite{Christensen2003}
of $\mathcal{L}_{2}(-f_{s}/2,f_{s}/2)$ by the Shannon-Nyquist sampling
theorem. For the class of sampling systems without preprocessing,
a fundamental rate limit necessary for perfect reconstruction of bandlimited
signals has been characterized by Landau using the definition of Beurling
density, as stated in the following theorem.

\begin{theorem}[\bf Landau Rate \cite{Landau1967}]\label{theorem:LandauRate}Consider
the set $\mathcal{B}_{\Omega}$ of all signals whose spectral contents
are supported on the frequency set $\Omega$. Suppose that pointwise
sampling without preprocessing is employed with a sampling set $\Lambda$.
If all signals $f(t)\in\mathcal{B}_{\Omega}$ can be uniquely determined
by the samples $\left\{ f(t_{n})\mid t_{n}\in\Lambda\right\} $, then
one must have $D^{-}(\Lambda)\geq\mu\left(\Omega\right)$. The value
$\mu\left(\Omega\right)$ is termed the Landau rate. \end{theorem}

Theorem \ref{theorem:LandauRate} characterizes the fundamental sampling
rate requirement for perfect signal reconstruction under pointwise
sampling without preprocessing. In particular, when $\Omega=[-B/2,B/2]$,
Theorem \ref{theorem:LandauRate} reduces to the Shannon-Nyquist theorem. 

We will thereby use Beurling density to characterize the sampling
rate for a general sampling set. However, since the preprocessor might
distort the time scale of the input, the resulting ``sampling rate''
might not characterize the true sampling rate applied to the original
signal, as illustrated in the following example.

\begin{example}[Compressor]Consider a preprocessing system defined
by the relation
\[
y(t)=\mathcal{T}\left(r(t)\right)=r\left(Lt\right)
\]
with $L\geq2$ being a positive integer. If we apply a uniform sampling
set $ $$\Lambda=\{t_{n}:t_{n}=n/f_{s}\}$ on the preprocessed output
$y(t)$, the sampled sequence at a ``sampling rate'' $f_{s}$ is
given by
\[
y[n]=y\left(n/f_{s}\right)=r\left(nL/f_{s}\right),
\]
which corresponds to sampling the system input $r(t)$ at rate $f_{s}/L$.
The compressor effectively time-warps the signal, thus resulting in
a mismatch of the time scales between the input and output.

\end{example}

The compressor example illustrates that the notion of sampling rate
may be misleading for systems that experience time warping. Hence,
this paper will focus only on sampling that preserves time scales.
One class of linear systems that preserves time scales are modulation
operators $\left(y(t)=p(t)x(t)\right)$, which perform pointwise scaling
of the input, and hence do not change the time scale. Another class
is the periodic system which includes LTI filtering, defined as follows. 

\begin{definition}[\bf Periodic System]A linear preprocessing system
is said to be periodic with period $T_{q}$ if its impulse response
$q(t,\tau)$ satisfies 
\begin{equation}
q(t,\tau)=q(t+T_{q},\tau+T_{q}),\quad\forall t,\tau\in\mathbb{R}.\label{eq:ShiftInvariantImpulseResponse}
\end{equation}
 \end{definition}

A more general class of systems that preserve the time scale can be
generated through modulation and periodic subsystems. Specifically,
we can define a general time-preserving system by connecting a set
of modulation or periodic operators in parallel or in serial. This
leads to the following definition.

\begin{definition}[\bf Time-preserving System]Given an index set
$\mathcal{I}$, a preprocessing system $\mathcal{T}:x(t)\mapsto\left\{ y_{k}(t),k\in\mathcal{I}\right\} $
is said to be time-preserving if

(1) The system input is passed through $\left|\mathcal{I}\right|$
(possibly countably many) branches of linear preprocessors, yielding
a set of analog outputs $\left\{ y_{k}(t)\mid k\in\mathcal{I}\right\} $.

(2) In each branch, the preprocessor comprises a set of periodic or
modulation operators connected in serial.

\end{definition}

With a preprocessing system that preserves the time scale, we can
now define the aggregate sampling rate through the Beurling density.

\begin{definition}[\bf Sampling Rate for Time-preserving Systems]A
sampling system is said to be time-preserving with sampling rate $f_{s}$
if

(1) Its preprocessing system $\mathcal{T}$ is time-preserving.

(2) The preprocessed output $y_{k}(t)$ is sampled by a sampling set
$\Lambda_{k}=\left\{ t_{l,k}\mid l\in\mathbb{Z}\right\} $ with a
uniform Beurling density $f_{k,\text{s}}$, which satisfies%
\footnote{We note that the sampling system may comprise countably many branches,
each with non-zero sampling rate. For instance, if the $k$th branch
is sampled at a rate $f_{k,\text{s}}=k^{-2}f_{0}$, we have an aggregate
rate $f_{s}=\sum_{k=1}^{\infty}k^{-2}f_{0}=\pi^{2}f_{0}/6$.%
} $\sum_{k\in\mathcal{I}}f_{k,\text{s}}=f_{s}$.\end{definition}

We note that the class of time-preserving sampling structures does
not preclude random sampling schemes. For example, the preprocessing
system can be a random modulator and the sampling set can be randomly
spaced. Our definition also includes multibranch sampling methods.
In fact, each multibranch sampling can be converted to an equivalent
single branch sampling as follows. 

\begin{prop}\label{factMultibranchSinglebranch}Suppose that a multibranch
sampling system has sampling rate $f_{s}$. Then there exists a single
branch sampling system with sampling rate $f_{s}$ that yields the
same set of sampled output values as the original system. This holds
simultaneously for all input signals. \end{prop}

\begin{proof}Suppose that the impulse response for the $k$th branch
is given by $q_{k}\left(t,\tau\right)$ with sampling set $\Lambda_{k}=\left\{ t_{k,n}\mid n\in\mathbb{Z}\right\} $.
Without loss of generality%
\footnote{In fact, if $\Lambda_{k}\cap\Lambda_{k'}\neq\emptyset$, then we can
introduce a new shifted pair $\left(q_{k}^{\delta}(t,\tau),\Lambda_{k}^{\delta}\right)$
such that $q_{k}^{\delta}\left(t+\delta,\tau\right):=q_{k}\left(t,\tau\right)$
and $\Lambda_{k}^{\delta}=\left\{ t+\delta\mid t\in\Lambda_{k}\right\} $
for some $\delta$ such that $\Lambda_{k}^{\delta}\cap\Lambda_{k}=\emptyset$,
i.e. we can introduce certain delay to the preprocessed output and
shift the sampling set correspondingly. Apparently, this new sampling
structure leads to the same collection of sample outputs. %
}, suppose that $\Lambda_{k}\cap\Lambda_{k'}=\emptyset$ for any $k\neq k'$.
By ordering all sample times in $\cup_{k\in\mathcal{I}}\Lambda_{k}$
and renaming them to be $\left\{ \tilde{t}_{l}\mid l\in\mathbb{Z}\right\} $
such that $\tilde{t}_{l}<\tilde{t}_{l+1}$ $ $for all $l$, we can
construct an equivalent single branch sampling system such that
\[
\tilde{q}\left(\tilde{t}_{l},\tau\right)=q_{k}\left(t_{k,n},\tau\right)
\]
if $\tilde{t}_{l}$ corresponds to $\tilde{t}_{k,n}$ in the original
sampling set. The sampling rate $\tilde{f}_{s}$ of the new system
is given by $\tilde{f}_{s}=\sum_{k\in\mathcal{I}}D(\Lambda_{k})=f_{s}$.\end{proof}

The samples obtained through this new single branch system preserve
all information we can obtain from the samples of the original multibranch
system. As will be seen, this proposition allows us to simplify the
analysis.

\subsection{Capacity Definition}

There are two capacity definitions that are of interest in sub-Nyquist
sampled channels: (1) the sampled capacity for a given sampling system;
(2) the capacity for a large class of sampling systems under a sampling
rate constraint. We now detail these definitions.

Suppose that the transmit signal $x(t)$ is constrained to the time
interval $\left[-T,T\right]$, and the received signal $y(t)$ is
sampled with sampling rate $f_{s}$ and observed over the time interval
$\left[-T,T\right]$. For a \emph{given} sampling system $\mathcal{P}$
that consists of a preprocessor $\mathcal{T}$ and a sampling set
$\Lambda$, and for a given time duration $T$, we define the information
metric $C_{T}^{\mathcal{P}}\left(P\right)$ to be 
\begin{equation}
C_{T}^{\mathcal{P}}\left(P\right)=\sup\frac{1}{2T}I\left(x\left(\left[-T,T\right]\right),\left\{ y[t_{n}]\right\} _{\left[-T,T\right]}\right),\label{eq:CapacityDefinitionFiniteT}
\end{equation}
 where the supremum is over all input distributions subject to a power
constraint $\mathbb{E}(\frac{1}{2T}\int_{-T}^{T}|x(t)|^{2}\mathrm{d}t)\leq P$.
Here, $\left\{ y[t_{n}]\right\} _{\left[-T,T\right]}$ denotes the
set of samples obtained at times within $\left[-T,T\right]\cap\Lambda$
by the sampling system $\mathcal{P}$, i.e. $\left\{ y[t_{n}]\mid n\in\mathbb{Z},t_{n}\in\left[-T,T\right]\right\} $. 

The capacity of the undersampled channel under a given sampling system
can then be studied by taking the limit as $T\rightarrow\infty$.
It was shown in \cite{ChenGolEld2010} that $\lim_{T\rightarrow\infty}C_{T}^{\mathcal{P}}\left(P\right)$
exists for a broad class of sampling methods, including sampling via
filter banks and sampling via periodic modulation. We caution, however,
that the existence of the limit is not guaranteed for all sampling
methods, e.g. the limit might not exist for irregular sampling. We
therefore define the capacity for a \emph{given} sampling system as
follows.

\begin{definition}$C^{\mathcal{P}}$ is said to be the\emph{ }information
capacity\emph{ }of a given sampled analog channel (or \emph{sampled
channel capacity}) if $\lim_{T\rightarrow\infty}C_{T}^{\mathcal{P}}(P)$
exists and 
\[
C^{\mathcal{P}}(P)=\lim_{T\rightarrow\infty}C_{T}^{\mathcal{P}}(P).
\]

\end{definition}

Note that any sampled analog channel can be converted to a set of
independent discrete channels via a Karhunen Loeve decomposition.
The metric $C^{\mathcal{P}}(P)$ then quantifi{}es asymptotically
the maximum mutual information between the input and output of these
discrete channels, or equivalently, the maximum data rate that can
be conveyed reliably through these channels.

The above capacity is defined for a given sampling mechanism. Another
metric of interest is the maximum data rate achievable by all sampling
schemes within a general class. This motivates us to define the sub-Nyquist
sampled channel capacity for a class of sampling systems as follows.

\begin{definition}[\bf Sampled Capacity under A Class of Sampling Systems]\label{defn:sampled_capacity_time_preserving}$C_{\mathcal{A}}(f_{s},P)$
is said to be the\emph{ }capacity\emph{ }of an analog channel over
all a class $\mathcal{A}$ of sampling systems under a given sampling
rate $f_{s}$ if 
\[
C_{\mathcal{A}}(f_{s},P)=\sup_{\mathcal{P}\in\mathcal{A}}C^{\mathcal{P}}(P).
\]

\end{definition}

The above definition of sub-sampled channel capacity characterizes
the capacity limit of an analog channel over a large set of sampling
mechanisms subject to a sampling rate constraint. This gives rise
to the natural problem of jointly optimizing the input and the sampling
architecture to maximize capacity, a goal we address in the next section.

\section{Main Results}

This section characterizes in closed form the sampled channel capacity
for a very general class of sampling systems, under a sampling rate
constraint. Specifically, we are concerned with the sampled channel
capacity $C_{\mathcal{A}}(f_{s},P)$, where 
\[
\mathcal{A}:=\left\{ \text{all right-invertible time-preserving sampling systems}\right\} .
\]
Here, the right-invertibility represents some mild regularity constraints
that ensure each sample contains innovation information, as will be
formally defined later. Unless otherwise specified, all sampling systems
mentioned in this section are assumed to be right-invertible time-preserving
linear systems.

Before proceeding, we shall assume, throughout, that for any frequency
$f$, the following holds
\begin{align}
\mathcal{S}_{\eta}\left(f\right) & \neq0,\nonumber \\
\int_{f}\frac{H^{2}\left(f\right)}{\mathcal{S}_{\eta}(f)}\mathrm{d}f & <\infty,\label{eq:AssumptionNoise}\\
\int_{f}\mathcal{S}_{\eta}(f)\mathrm{d}f & <\infty\text{ or }\mathcal{S}_{\eta}\left(f\right)\equiv1.\nonumber 
\end{align}

\subsection{An Upper Bound on Sampled Channel Capacity\label{sub:CapacityUpperBound}}

\subsubsection{The Converse\label{sub:The-Converse}}

A time-preserving sampling system preserves the time scale of the
signal, and hence does not compress or expand the frequency scale.
We now determine an upper limit on the sampled channel capacity for
this class of general nonuniform sampling systems. Proposition \ref{factMultibranchSinglebranch}
implies that any multibranch sampling system can be converted into
a single branch sampling system without loss of information. Therefore,
we restrict our analysis in this section to the class of single branch
sampling systems, which provides exactly the same upper bound as the
one accounting for multibranch systems. In addition, we constrain
our attention to sampling methods that are right-invertible, as defined
below%
\footnote{We impose the right-invertibility constraint primarily for the sake
of mathematical convenience. We conjecture, however, that removing
this constraint does not change our main result (Theorem \ref{thm:GeneralSampledCapacity}). %
}.

\begin{definition}[\bf Right-Invertible Sampling System]A sampling
system with sampling set $\Lambda$ and impulse response $q(t_{i},\tau)$
($t_{i}\in\Lambda$) is said to be right-invertible with respect to
$\mathcal{S}_{\eta}(f)$ if 
\begin{enumerate}
\item for any $k\in\mathbb{Z}$, the frequency response $\mathcal{F}\left(q_{k}\left(\tau\right)\right)$
is bounded;
\item for any frequency $f$ and any $T$ with its associated sampling subset
\[
\Lambda_{\left[-T,T\right]}=\left[-T,T\right]\cap\Lambda:=\left\{ t_{1},\cdots,t_{N_{T}}\right\} ,
\]
the least singular value of the $N_{T}\times\infty$-dimensional matrix
$\boldsymbol{F}_{T}\left(f\right)$ is uniformly bounded away from
0. 
\end{enumerate}
Here, $q_{k}\left(\tau\right):=q\left(t_{k},t_{k}-\tau\right)$, and
$\boldsymbol{F}_{T}$ is a $N_{T}\times\infty$ matrix defined by
\[
\left[\boldsymbol{F}_{T}\left(f\right)\right]_{k,l}=\mathcal{F}\left(s_{\eta}\cdot q_{k}\right)\left(f+\frac{l}{T}\right),\text{ }\text{ }1\leq k\leq N_{T},l\in\mathbb{Z}.
\]

\end{definition}

The right invertibility of the sampling system essentially implies
that for each subset of the impulse response $\left\{ q(t_{i},\tau)\mid i\in\mathcal{I}\right\} $,
the Fourier matrix associated with any sampled response is right-invertible.
This typically implies that the set of samples is a linearly independent
family -- each sample provides a sufficient amount of innovative information.
Our main theorem is now stated as follows.

\begin{theorem}[\bf Converse]\label{thm:GeneralSampledCapacity}Assume
that there exists a small constant $\epsilon>0$ such that $\mathcal{F}^{-1}\left(\frac{H(f)}{\sqrt{\mathcal{S}_{\eta}(f)}}\right)(t)=O\left(\frac{1}{t^{1.5+\epsilon}}\right)$.
Suppose that there exists a frequency set $B_{\mathrm{m}}$ with $\mu\left(B_{\mathrm{m}}\right)=f_{s}$
that satisfies
\[
{\displaystyle \int}_{f\in B_{\mathrm{m}}}\frac{\left|H(f)\right|^{2}}{\mathcal{S}_{\eta}(f)}\mathrm{d}f=\sup_{B:\mu\left(B\right)=f_{s}}{\displaystyle \int}_{f\in B}\frac{\left|H(f)\right|^{2}}{\mathcal{S}_{\eta}(f)}\mathrm{d}f.
\]

Under any time-preserving right-invertible sampling system $\mathcal{P}$
w.r.t. $\mathcal{S}_{\eta}\left(f\right)$ with sampling rate $f_{s}$,
the sampled channel capacity is upper bounded by
\begin{align}
C^{\mathcal{P}}\left(P\right) & \leq C_{\mathrm{u}}\left(f_{s},P\right)\nonumber \\
 & :={\displaystyle \int}_{f\in B_{\mathrm{m}}}\frac{1}{2}\left[\log\left(\nu\frac{\left|H(f)\right|^{2}}{\mathcal{S}_{\eta}(f)}\right)\right]^{+}\mathrm{d}f,\label{eq:GeneralCapacityUpperBound}
\end{align}
where $\nu$ is given parametrically by
\begin{align}
{\displaystyle \int}_{f\in B_{\mathrm{m}}}\left[\nu-\frac{\mathcal{S}_{\eta}(f)}{\left|H(f)\right|^{2}}\right]^{+}\mathrm{d}f & =P.\label{eq:GeneralCapacityWaterLevel}
\end{align}

\end{theorem}

\begin{remark}Note that $C_{\mathrm{u}}\left(f_{s},P\right)$ is
monotonically nondecreasing in $f_{s}$ and $P$. In fact, when the
sampling rate is increased from $f_{s}$ to $f_{s}+\delta$, $C_{\mathrm{u}}\left(f_{s},P\right)$
corresponds to the optimal value when considering all spectral sets
of support size $f_{s}+\delta$. Since we are still allowed to employ
(suboptimal) strategies to allocate power over smaller spectral sets
with size $f_{s}$, we are optimizing over a larger set of transmission
/ power allocation strategies than the situation with sampling rate
$f_{s}$. Therefore, $C_{\text{u}}\left(f_{s},P\right)$ is nondecreasing
in $f_{s}$.\end{remark}

It can be easily shown that the upper limit $C_{\mathrm{u}}\left(f_{s},P\right)$
is equivalent to the maximum capacity of a channel whose spectral
occupancy is no larger than $f_{s}$. The above result basically implies
that even if we allow for more complex irregular sampling sets, the
sampled capacity cannot exceed the one commensurate with the analog
capacity when constraining all transmit signals to the interval of
bandwidth $f_{s}$ that experience the highest SNR. Accordingly, the
optimal input distribution will lie in this maximizing frequency set.
This theorem also demonstrates that the capacity is attained when
aliasing is suppressed by the sampling structure, as will be seen
later in our capacity-achieving scheme. When the optimal frequency
interval $B_{\text{m}}$ is selected, a water filling power allocation
strategy is performed over the spectral domain with some water level
$\nu$ determined by (\ref{eq:GeneralCapacityWaterLevel}).

\subsubsection{Approximate Analysis\label{sub:Approximate-Analysis}}

To give some intuition into the results, we provide an approximate
(but non-rigorous) argument relying on ``noise whitening'' and ``orthonormal
projections''.

Suppose that the Fourier transform of the analog channel output $r(t)$
is given by $H(f)X(f)+N(f)$, where $X(f)$ and $N(f)$ denote, respectively,
the frequency responses of $x(t)$ and $\eta(t)$. When the sampled
sequence does not collapse information, we can characterize the sampling
process through a linear injective mapping $\mathcal{R}$ from the
space of linear functions $H(f)X(f)+N(f)\in\mathcal{L}_{2}(-\infty,\infty)$
onto the space $\mathcal{L}_{2}(-f_{s}/2,f_{s}/2)$ of bandlimited
functions:
\[
\phi\left(\cdot\right)=\mathcal{R}\left(HX\right)+\mathcal{R}\left(N\right).
\]
This way the noise component $\mathcal{R}\left(N\right)$ can be treated
as additive sampled noise in the frequency domain. We note, however,
that this Gaussian noise $\mathcal{R}\left(N\right)$ is not necessarily
independent over the spectral support $\left[-f_{s}/2,f_{s}/2\right]$.
This motivates us to whiten it first without loss of information. 

Denote by $\mathcal{W}$ the whitening operator and suppose that $\mathcal{R}\left(N\right)$
is bounded away from 0. Then the prewhitening process is performed
as 
\[
\mathcal{W}\phi\left(\cdot\right)=\mathcal{W}\left(\mathcal{R}\left(HX\right)\right)+\mathcal{W}\left(\mathcal{R}\left(N\right)\right)
\]
such that the noise component $\mathcal{W}\left(\mathcal{R}\left(N\right)\right)$
is independent across the frequency domain. If we set $\tilde{\mathcal{R}}(\cdot)\overset{\Delta}{=}\mathcal{W}\left(\mathcal{R}\left(\cdot\right)\right)$,
then we can rewrite the input-output relation as
\[
\tilde{\phi}\left(\cdot\right)=\tilde{\mathcal{R}}\left(HX\right)+\tilde{N},
\]
with $\tilde{N}$ being white over $\left[-f_{s}/2,f_{s}/2\right]$
and $\tilde{\mathcal{R}}$ being an \emph{orthonormal} operator onto
$\mathcal{L}_{2}\left(-f_{s}/2,f_{s}/2\right)$. That said, the operator
$\tilde{\mathcal{R}}$ effectively projects all spectral components
$H(f)X(f)+N(f)$ onto a subspace $\mathcal{L}_{2}(-f_{s}/2,f_{s}/2)$.
Instead of scrambling the spectral contents, the optimal projection
that maximizes the SNR extracts out a spectral set $B_{\text{m}}$
of support size $f_{s}$ that contains the frequency components with
the highest SNR. This leads to the capacity upper bound (\ref{eq:GeneralCapacityUpperBound}).
As illustrated in Fig. \ref{fig:ProjectionSpectralContents}, scrambling
the spectral contents does not in general improve capacity. This will
be formally proved in Appendix \ref{sec:Proof-Architecture-of-Theorem-General-Capacity}.

\begin{figure}
\begin{centering}
\emph{}%
\begin{tabular}{cc}
\includegraphics[scale=0.3]{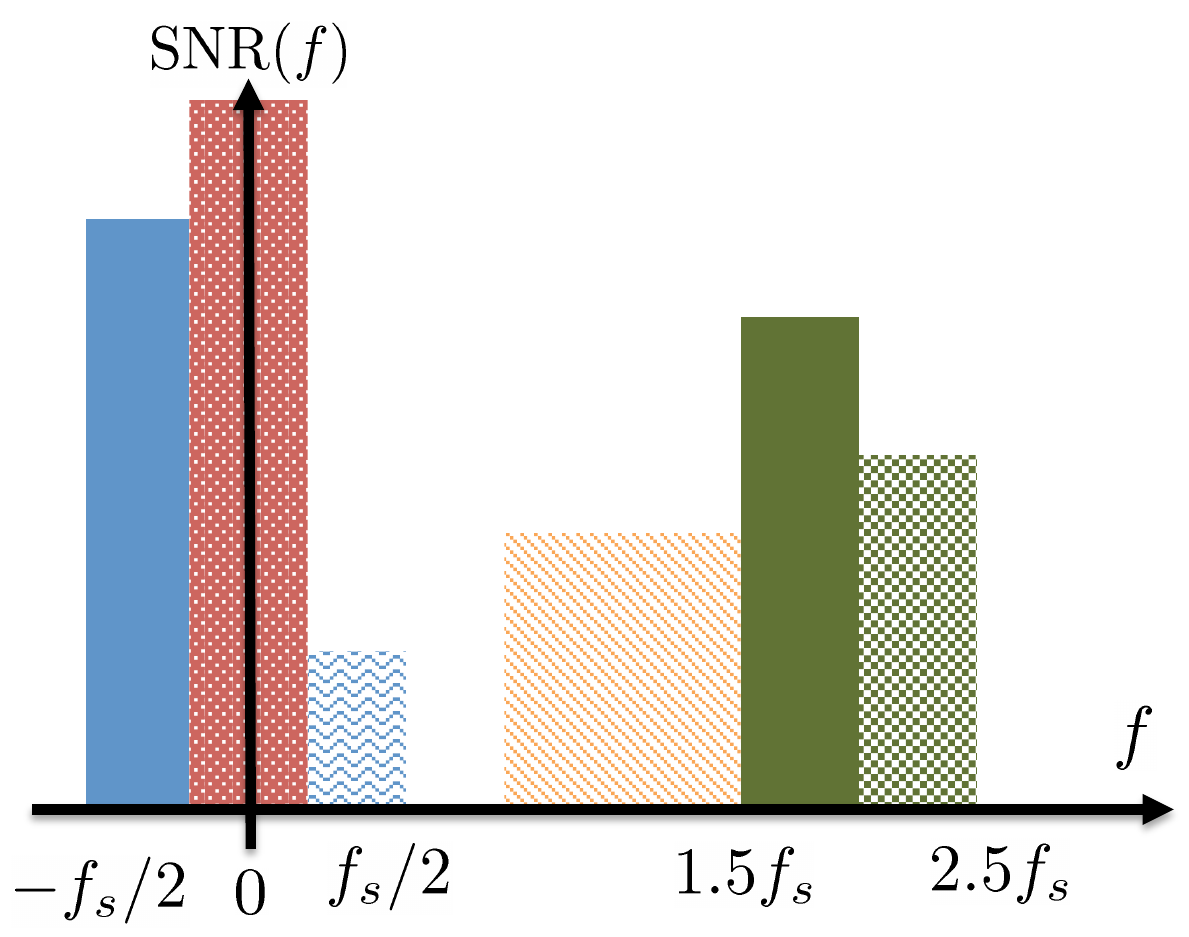} & \includegraphics[scale=0.3]{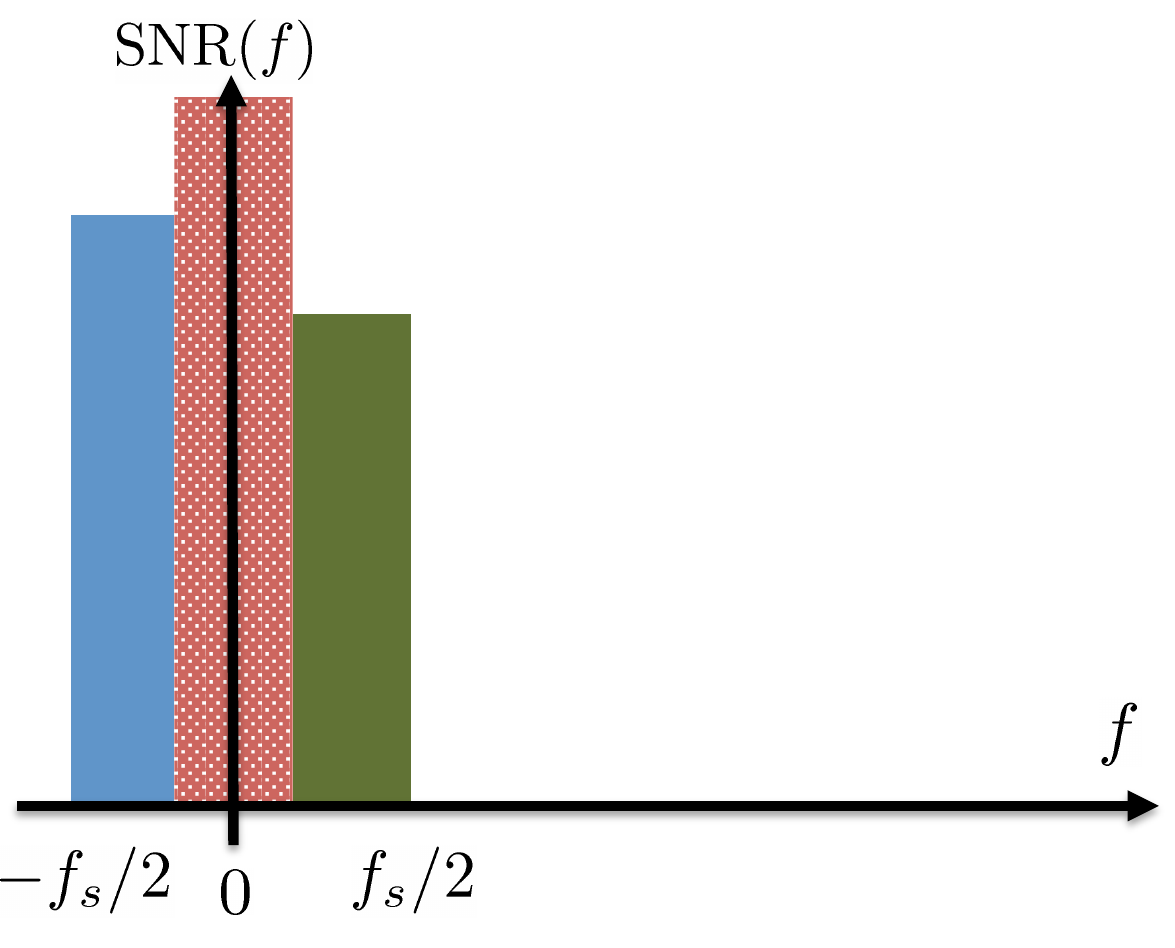}\tabularnewline
(a) & (b)\tabularnewline
\end{tabular}
\par\end{centering}

\begin{centering}
\emph{}%
\begin{tabular}{c}
\emph{\includegraphics[scale=0.3]{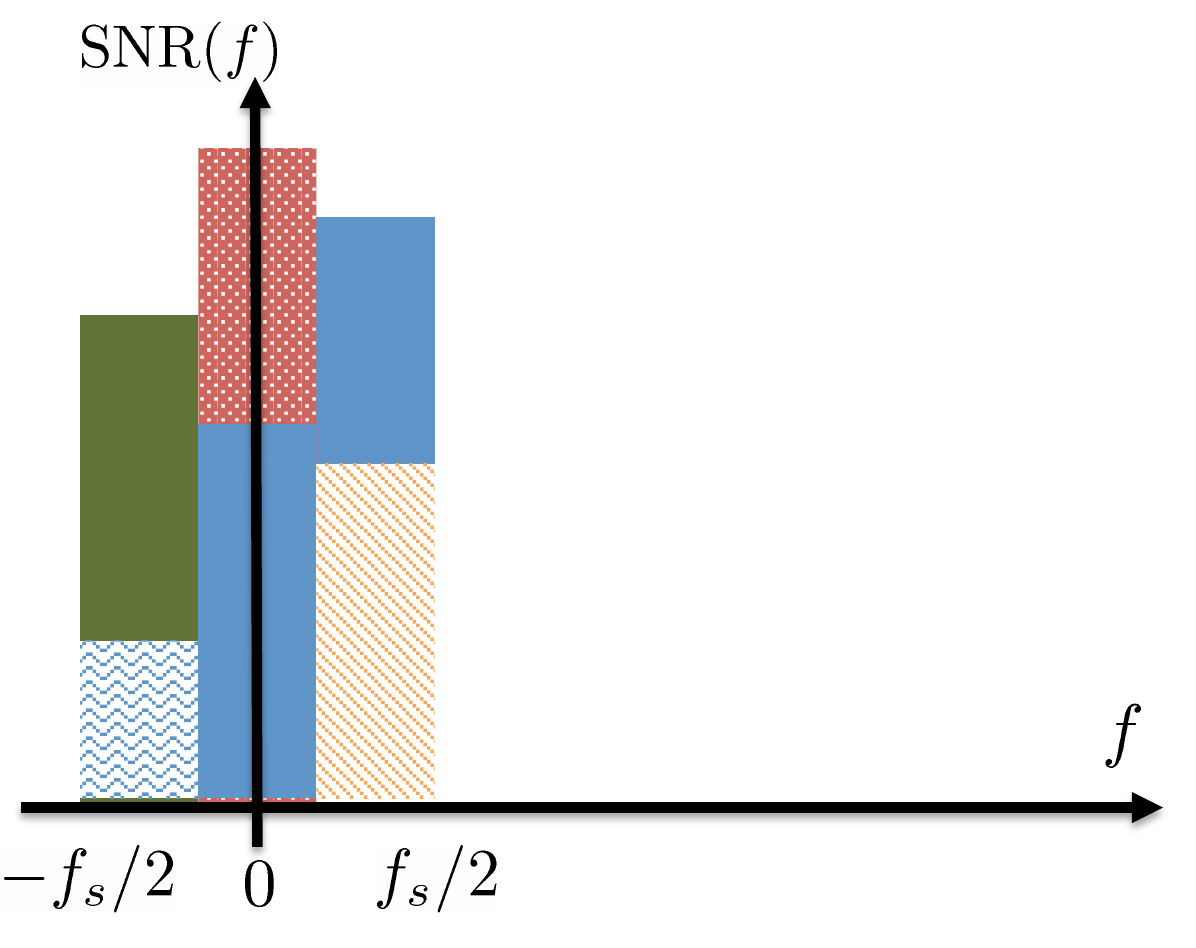}}\tabularnewline
(c)\tabularnewline
\end{tabular}
\par\end{centering}

\caption{\label{fig:ProjectionSpectralContents}Projection of spectral contents
from $\mathcal{L}\left(-\infty,-\infty\right)$ onto $\mathcal{L}\left(-f_{s}/2,f_{s}/2\right)$.
(a) SNR of the analog channel. (b) optimal projection: it extracts
out a frequency set of size $f_{s}$ and zeros out all other contents.
(c) a projection that scrambles the spectral contents, which does
not in general improve capacity. }
\end{figure}

\subsubsection{Proof Sketch\label{sub:ProofOutline}}

We now outline the key steps underlying the proof of Theorem \ref{thm:GeneralSampledCapacity}
for the white-noise scenario. 

i) We start by analyzing the class of periodic sampling systems: a
special type of sampling methods that allow closed-form capacity expressions.
We then demonstrate that the capacity under any periodic sampling
system with sampling rate $f_{s}$ and transmit power $P$ is bounded
above by $C_{\text{u}}\left(f_{s},P\right)$.

ii) The upper bound is then derived by relating general (possibly
aperiodic) sampled channels with periodic sampled channels through
a \emph{finite-duration} approximation argument. In fact, instead
of studying the true sampled channel response directly, we truncate
the channel response so that its impulse response is nonzero only
for a finite duration. The capacity bound for the resulting truncated
channel is then bounded by the capacity of a new periodized channel
we construct.  As we show, the capacity of the truncated channel can
be made arbitrarily close to the capacity of the true sampled channel. 

The most technically involved step is Step ii), which proceeds by
considering two cases as follows. 
\begin{enumerate}
\item \textbf{Finite-duration} $h(t)$. Consider first channels for which
$h(t)$ is of finite duration, $h(t)=0$ for any $t\notin[-L_{0},L_{0}]$
for some $L_{0}>0$.

$\quad$(a) Consider any given right-invertible time-preserving sampling
system $\mathcal{P}$ with impulse response $q(t,\tau)$, and suppose
that the input $x(t)$ is time constrained to the interval $[-T,T]$.
Construct a periodic channel with period $2(T+L_{0})$ based on $q(t,\tau)$.
Let $C_{\text{p}}^{\mathcal{P}}\left(P\right)$ denote the capacity
of the periodized channel, whose sampling rate is bounded above by
$f_{s}+\epsilon$ for some arbitrarily small $\epsilon>0$.

$\quad$(b) Show that $C_{T}^{\mathcal{P}}\left(P\right)\leq\frac{T+L_{0}}{T}C_{\text{p}}^{\mathcal{P}}\left(\frac{T}{T+L_{0}}P\right)$
holds uniformly for all $\mathcal{P}$. Since we know that $C_{\text{p}}^{\mathcal{P}}\left(P\right)\leq C_{\text{u}}\left(f_{s}+\epsilon,P\right)$
for any periodized channel (or, equivalently, any channel followed
by a periodic sampling system), this establishes the capacity upper
bound for this class of finite-duration channels, provided that $T$
is sufficiently large.

\item \textbf{Infinite-duration} $h(t)$. We next extend the results to
channels for which $h(t)$ is non-zero over infinite duration.

$\quad$(a) Construct a truncated channel such that 
\[
\tilde{h}(t)=\begin{cases}
h(t),\quad & \text{if }\left|t\right|\leq L_{1},\\
0, & \text{else},
\end{cases}
\]
for some sufficiently large $L_{1}$. The capacity upper bound holds
for the truncated channel, as shown in Step 1).

$\quad$(b) For any given sampling system $\mathcal{P}$ and any time
interval $[-T,T]$, compare the capacity of the original channel (denoted
by $C_{T}^{\mathcal{P}}\left(P\right)$) with the capacity of the
truncated channel (denoted by $\tilde{C}_{T}^{\mathcal{P}}\left(P\right)$),
which can be completed by investigating the spectrum of the operators
associated with both sampled channels. It can be shown that $C_{T}^{\mathcal{P}}\left(P\right)$
can be upper bounded by $\tilde{C}_{T}^{\mathcal{P}}\left(P+\xi\right)+\xi$
for some arbitrarily small constant $\xi>0$, which holds uniformly
over all sampling systems $\mathcal{P}$. Combining this with results
shown in Step 1), we demonstrate that $\tilde{C}_{T}^{\mathcal{P}}\left(P\right)$
(and hence $\tilde{C}_{T}^{\mathcal{P}}\left(P\right)$) is bounded
arbitrarily close by $C_{\mathrm{u}}\left(f_{s},P\right)$ , which
establishes the claim for the whole class of infinite-duration channels.

\end{enumerate}

\subsection{Achievability\label{sub:Achievability}}

It turns out that for most scenarios of interest, the capacity upper
bound given in Theorem \ref{thm:GeneralSampledCapacity} can be attained
through filterbank sampling, as stated in the following theorem.

\begin{theorem}[\bf Achievability -- Sampling with a Filter Bank]\label{thm:OptimalSamplingGeneralSampledCapacity}Suppose
that the maximizing frequency set $B_{\mathrm{m}}$ introduced in
Theorem \ref{thm:GeneralSampledCapacity} exists and is piecewise
continuous or, more precisely, 
\[
B_{\mathrm{m}}=\cup_{i}B_{i\in\mathcal{X}}\text{ },
\]
where $\mathcal{X}$ is an index set, and $B_{i}$'s are some non-overlapping
continuous intervals. Consider the following filterbank sampling mechanism
$\mathcal{P}_{\mathrm{FB}}$: in the $k$th branch, the frequency
response of the filter is given by
\begin{equation}
S_{k}(f)=\begin{cases}
1,\quad & \text{if }f\in B_{k},\\
0, & \text{otherwise},
\end{cases}\label{eq:OptimalFilterBank}
\end{equation}
and each filter is followed by an ideal uniform sampler with sampling
rate $\mu\left(B_{k}\right)$. Then
\[
C^{\mathcal{P}_{\mathrm{FB}}}\left(P\right)=C_{\mathrm{u}}\left(f_{s},P\right),
\]
where $C_{\mathrm{u}}\left(f_{s},P\right)$ is the upper bound given
by (\ref{eq:GeneralCapacityUpperBound}). \end{theorem}\begin{IEEEproof}The
spectral components in $B_{i}$ can be perfectly reconstructed from
the sequence that is obtained by first extracting out a subinterval
$B_{i}$ and then uniformly sampling the filtered output with sampling
rate $f_{s,i}$. The capacity under $\mathcal{P}_{\mathrm{FB}}$ is
commensurate to the analog capacity when constraining the transmit
signal to $\cup_{i}B_{i}$, which is equivalent to $C_{\mathrm{u}}\left(f_{s},P\right)$.
\end{IEEEproof}

Note that the bandwidth of $B_{i}$ may be irrational and the system
may require an infinite number of filters. Theorem \ref{thm:OptimalSamplingGeneralSampledCapacity}
indicates that filterbank sampling with \emph{varied sampling rates}
in different branches maximizes capacity. $ $

\begin{figure}[htbp]
\begin{centering}
\textsf{\includegraphics[scale=0.38]{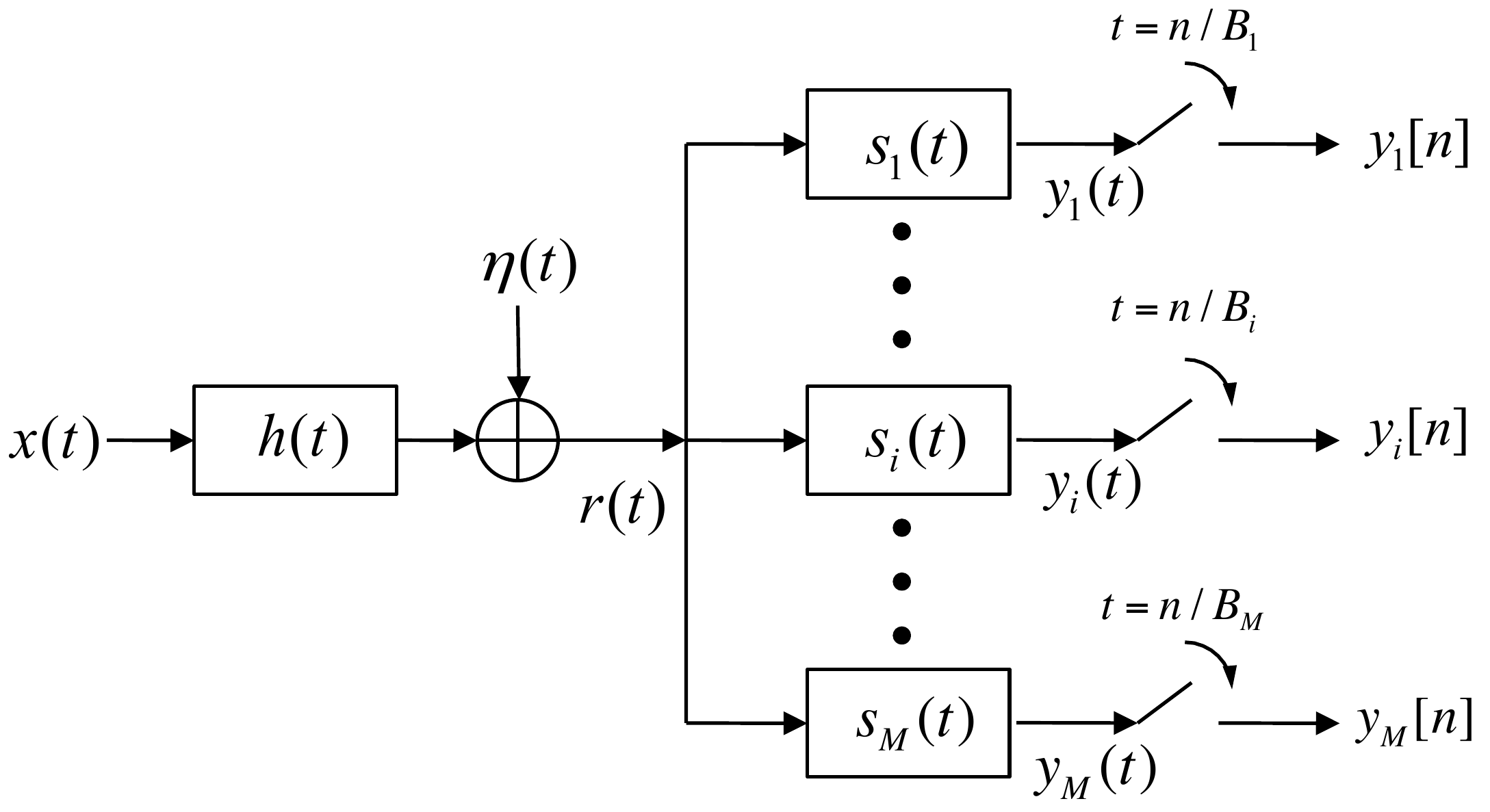}} 
\par\end{centering}

\begin{centering}
(a) 
\par\end{centering}

\begin{centering}
\textsf{\includegraphics[scale=0.45]{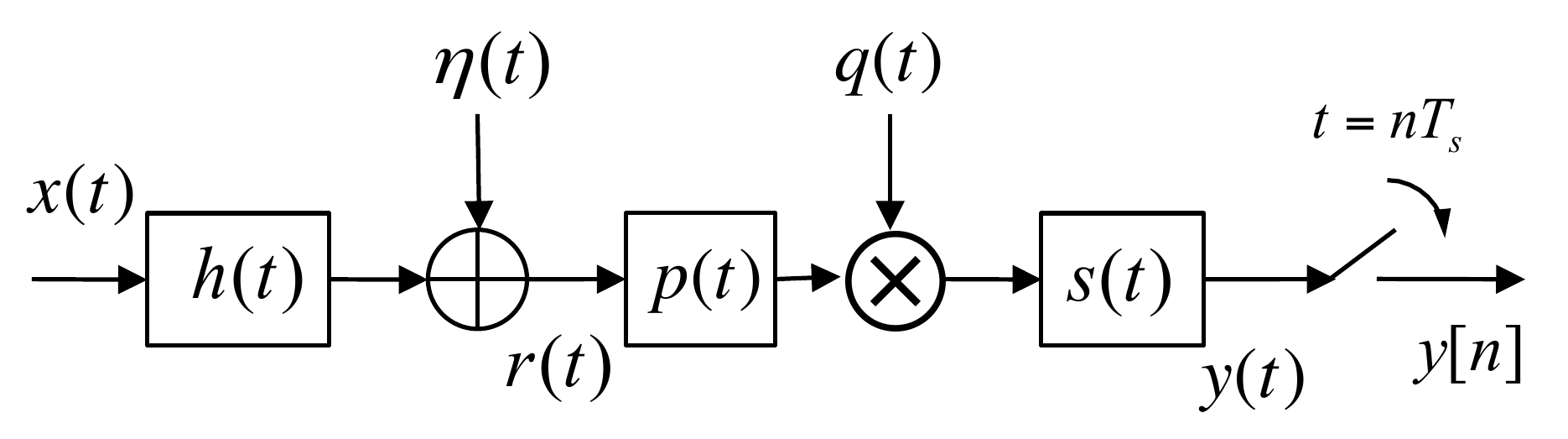}} 
\par\end{centering}

\begin{centering}
(b) 
\par\end{centering}

\caption{\label{fig:FilterBankModulation}(a) Filterbank sampling: each branch
filters out a frequency interval of bandwidth $B_{k}$, and samples
it with rate $f_{k,\text{s}}=B_{k}$; (b) A single branch of modulation
and filtering: the channel output is prefiltered by a filter with
impulse response $p(t)$, modulated by a sequence $q(t)$, post-filtered
by another filter of impulse response $s(t)$, and finally sampled
by a uniform sampler at a rate $f_{s}$. If the SNR $|H(f)|^{2}/\mathcal{S}_{\eta}(f)$
is piecewise flat, then $p(t)$, $q(t)$ and $s(t)$ can be chosen
such that the two systems are equivalent in terms of sampled capacity.}
\end{figure}

The optimality of filterbank sampling immediately leads to another
optimal sampling structure under mild conditions. As we have shown
in \cite{ChenGolEld2010}, filterbank sampling with equal rates on
different branches is equivalent to a single branch of modulation,
as illustrated in Fig. \ref{fig:FilterBankModulation}. This approach
attains the sampled capacity achievable by filterbank sampling if
the SNRs of the analog channel are piecewise constant in frequency.
Although the filterbank sampling we derive in (\ref{eq:OptimalFilterBank})
does not employ equal rates on different branches, for most channels
of physical interest we can simply divide each branch further into
a number of sub-branches to allow the rates at different branches
to be reasonably close to each other. Therefore, for most channels
of physical interest (say, the channels whose SNRs in frequency are
Riemann-integrable), the capacity achievable through filterbank sampling
can be approached arbitrarily closely by a single branch of sampling
with modulation. This achievability result is formally stated in the
following theorem.

\begin{theorem}[\bf Achievability -- A Single Branch of Sampling with Modulation and Filtering]\label{theorem:AchievabilityModulation}Under
the assumptions of Theorem \ref{thm:OptimalSamplingGeneralSampledCapacity},
suppose further that $\left|H(f)\right|^{2}/\mathcal{S}_{\eta}(f)$
is constant within each set $B_{i}$. Then for any $\epsilon>0$,
there exists a time-preserving sampling system $\mathcal{P}_{\text{MF}}$
with sampling rate $f_{s}$ using a single branch of sampling with
modulation and filtering such that $C^{\mathcal{P}_{\text{MF}}}\left(P\right)\geq C_{\mathrm{u}}\left(f_{s},P\right)-\epsilon$,
where $C_{\mathrm{u}}\left(f_{s},P\right)$ is defined in (\ref{eq:GeneralCapacityUpperBound}).

\end{theorem}

\begin{IEEEproof}It is straightforward to see that there exists a
set of non-overlapping intervals $\left\{ \tilde{B}_{k}\right\} $
each with \emph{equal measure} that can approximate the original sets
$\left\{ B_{l}\mid1\leq l\leq n\right\} $ arbitrarily well. Employing
the sampling method described in \cite[Section V.D]{ChenGolEld2010}
achieves a sampled capacity arbitrarily close to $C_{\mathrm{u}}\left(f_{s},P\right)$.\end{IEEEproof}

A channel of physical interest can often be approximated as piecewise
constant over frequency in this way. Given the maximizing frequency
set $B_{\text{m}}$, the sampling structure suggested in \cite[Section V.D]{ChenGolEld2010}
first suppresses the frequency components outside $B_{\text{m}}$
using an optimal LTI prefilter. A modulation module is then applied
to scramble all frequency components within $B_{\text{m}}$. The aliasing
effect can be significantly mitigated by appropriate choices of modulation
weights for different spectral subbands. We then employ another band-pass
filter to suppress out-of-band signals, and sample the output using
a pointwise uniform sampler. Compared with filterbank sampling, a
single branch of modulation and filtering only requires the design
of a lowpass filter, a band-pass filter, and a multiplication module,
which might be of lower complexity to implement than a filter bank.

\section{Discussion\label{sec:Discussion}}

Some properties of the capacity and capacity-achieving strategies
are now discussed.
\begin{itemize}
\item \textbf{Monotonicity}. It can be seen from Theorem \ref{thm:GeneralSampledCapacity}
that increasing the sampling rate from $f_{s}$ to $\tilde{f_{s}}$
results in another frequency set $\tilde{B}_{\text{m}}$ of support
size $\tilde{f}_{s}$ that has the highest SNRs. By definition, the
original frequency set $B_{\text{m}}$ must be a subset of $\tilde{B}_{\text{m}}$.
Therefore, the sampled capacity with rate $\tilde{f}_{s}$ is no lower
than the sampled capacity with rate $f_{s}$.
\item \textbf{Irregular sampling} \textbf{set}. Sampling with irregular
nonuniform sampling sets, while requiring complicated reconstruction
and interpolation techniques \cite{AldGro2001}, does not outperform
filterbank or modulation bank sampling with regular uniform sampling
sets in maximizing capacity for the channels considered herein. 
\item \textbf{Alias suppression}. We have seen that aliasing does not allow
a higher capacity to be achieved when perfect channel state information
is known at both the transmitter and the receiver. The optimal sampling
method corresponds to the optimal alias-suppression strategy. This
is in contrast to the benefits obtained through random mixing of spectral
components in many sub-Nyquist sampling schemes with unknown signal
supports. When we are allowed to jointly optimize over both input
and sampling schemes with perfect channel state information, scrambling
of spectral contents does not in general maximize capacity.
\item \textbf{Perturbation of the sampling} \textbf{set}. If optimal filterbank
or modulation sampling is employed, then mild perturbation of post-filtering
uniform sampling sets does not degrade the sampled capacity. One surprisingly
general example was proved by Kadec \cite{Kadec1964}. Suppose that
a sampling rate $\hat{f}_{s}$ is used in any branch and the sampling
set satisfies $\left|\hat{t}_{n}-n/\hat{f}_{s}\right|\leq\hat{f}_{s}/4$.
Then $\left\{ \exp\left(j2\pi\hat{t}_{n}f\right)\mid n\in\mathbb{Z}\right\} $
also forms a Riesz basis of $\mathcal{L}_{2}(-\hat{f}_{s}/2,\hat{f}_{s}/2)$,
thereby preserving information integrity. These nonuniform sampling
and reconstruction schemes, while generally complicated to implement
in practice, significantly broaden the class of sampling mechanisms
that allow perfect reconstruction of bandlimited signals, and indicate
stability and robustness of the sampling sets. Kadec's result immediately
implies that the sampled capacity is invariant under mild perturbation
of the sampling sets.
\item \textbf{Hardware implementation}. When the sampling rate is increased
from $f_{s1}$ to $f_{s2}$, we need only to insert an additional
filter bank of overall sampling rate $f_{s2}-f_{s1}$ to extract out
another set of spectral components with bandwidth $f_{s2}-f_{s1}$.
Thus, the adjustment of the sampling hardware system for filterbank
sampling is incremental with no need to rebuild the whole system from
scratch.
\item \textbf{Spectrum Blind Sampling}. This paper focuses on the scenario
with perfect channel state information known at the transmitter, the
receiver, and the sampler. This is different from the setting of compressed
sensing, where the signal spectrum is unknown to the sampler and the
decoder. In fact, the alias-suppressing sampler requires knowledge
of the channel. If this knowledge is not available, then alias-suppressing
samplers might result in low capacity. When the sampler is spectrum
blind and the channel realization is uncertain, random sampling that
scrambles the spectral contents \cite{MisEld2010Theory2Practice,GedTurEld2010}
outperforms alias-suppressing sampling in minimizing the rate loss
due to channel-independent sampling design. We investigate the capacity
of sub-Nyquist sampled channels with unknown CSI in our companion
paper \cite{ChenEldarGoldsmith2013minimax}. 
\end{itemize}

\section{Concluding Remarks\label{sec:Concluding-Remarks}}

We developed the maximum achievable information rate for a general
class of right-invertible time-preserving nonuniform sampling methods
under a sampling rate constraint. It is shown that nonuniformly spaced
sampling sets, while requiring fairly complicated reconstruction /
approximation algorithms, do not provide any capacity gain. Encouragingly,
filterbank sampling with varied sampling rates on different branches,
or a single branch of sampling with modulation and filtering, are
sufficient to achieve the sampled channel capacity. In addition, both
strategies suppress aliasing effects. In terms of maximizing capacity,
there is no need to employ irregular sampling sets that are more complicated
to implement in practical hardware systems. The resulting sampled
capacity is shown to be monotonically increasing in sampling rate.

Our results in this paper are based on the assumption that perfect
channel state information is known at both the transmitter and the
receiver. It remains to be seen what sampling strategies can optimize
information rates when only partial channel state information is known.
It is unclear whether anti-aliasing methods are still optimal in maximizing
capacity. Moreover, when it comes to the multi-user information theory
setting, anti-aliasing methods might not outperform other spectral-mixing
approaches in the entire capacity region. It would be interesting
to see how to optimize the sampling schemes in multi-user channels,
for example, joint sampling schemes in sampled multiple access analog
channels.

\section*{Acknowledgments}

The authors would like to thank Prof. Young-han Kim and the anonymous
reviewers for extensive and constructive comments that greatly improved
the paper. 

\appendices

\section{Proof of Theorem \ref{thm:GeneralSampledCapacity}\label{sec:Proof-Architecture-of-Theorem-General-Capacity}}

For simplicity of presentation, we assume throughout that the noise
is white, i.e. $S_{\eta}\left(f\right)\equiv1$. In fact, under the
assumption (\ref{eq:AssumptionNoise}), we can always split the channel
filter $H\left(f\right)$ into two parts with respective frequency
response $H\left(f\right)/\sqrt{\mathcal{S}_{\eta}(f)}$ and $\sqrt{\mathcal{S}_{\eta}(f)}$.
Since the colored noise is equivalent to a white Gaussian noise passed
through a filter with transfer function $\sqrt{\mathcal{S}_{\eta}(f)}$,
the original system can be redrawn as in Fig. \ref{fig:FiniteDurationColorNoise}.
The filter with frequency response $\sqrt{\mathcal{S}_{\eta}(f)}$
can then be incorporated into the preprocessing system to generate
a new time-preserving preprocessor.

\begin{figure}[htbp]
\begin{centering}
\textsf{\includegraphics[scale=0.4]{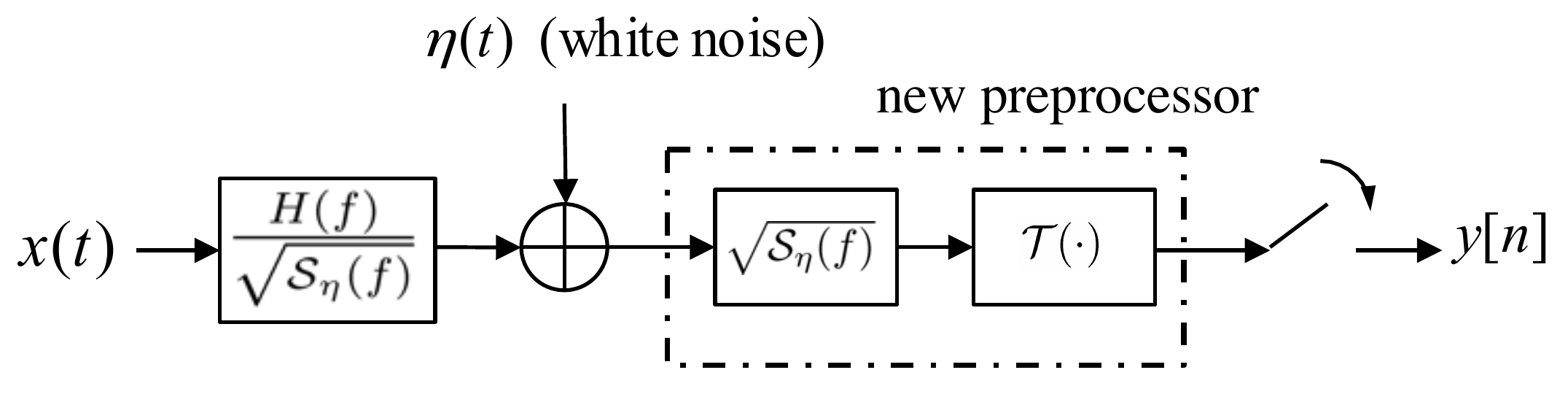}} 
\par\end{centering}

\caption{\label{fig:FiniteDurationColorNoise}Equivalent representation of
sampling systems in the presence of colored noise.}
\end{figure}

Recall that Proposition \ref{factMultibranchSinglebranch} indicates
that any multibranch sampling system can be converted into a single
branch sampling system without loss of information. As a result, we
restrict our proof to the class of single branch sampling systems.

\subsection{Capacity under Periodic Sampling Systems}

Recall that for a time varying system, the impulse response $q(t,\tau)$
is defined as the output seen at time $t$ due to an impulse in the
input at time $\tau$. The sampling system may not be time-invariant,
but a broad class of sampling mechanisms applied in practice exhibit
block-wise time invariance properties. Specifically, we introduce
the notion of periodic sampling systems as follows.

\begin{definition}[{\bf Periodic Sampling}]\label{definition-PeriodicSampling}Consider
a sampling system with a preprocessing system of impulse response
$q(t,\tau)$ followed by a sampling set $\Lambda=\left\{ t_{k}\mid k\in\mathbb{Z}\right\} $.
A linear sampling system is said to be periodic with period $T_{q}$
and sampling rate $f_{s}$ ($f_{s}T_{q}\in\mathbb{Z}$) if the preprocessing
system is periodic with period $T_{q}$ and the sampling set satisfies
\begin{align}
t_{k+f_{s}T_{q}} & =t_{k}+T_{q},\quad\forall k\in\mathbb{Z}.\label{eq:PeriodicSamplingSequence}
\end{align}
\end{definition}

\begin{figure*}[tbph]
\begin{centering}
\textsf{\includegraphics[scale=0.48]{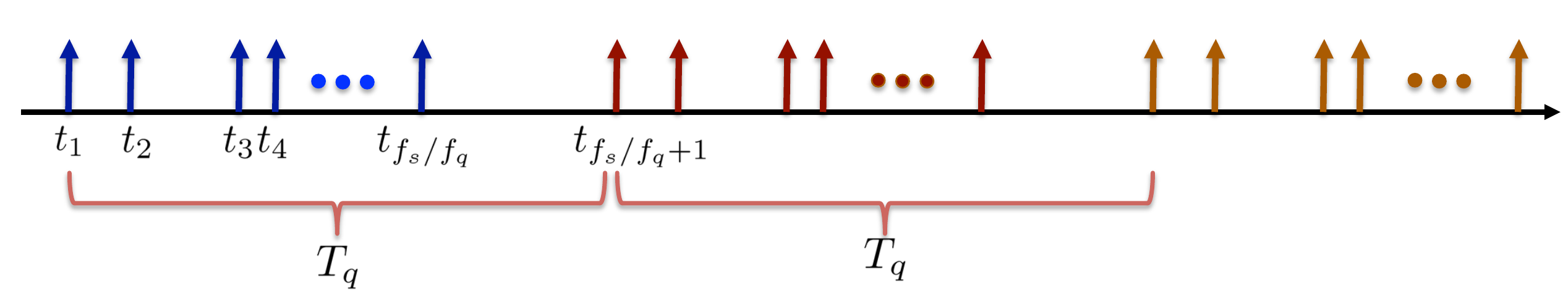}} 
\par\end{centering}

\caption{\label{fig:PeriodicSamplingGrid} The sampling set of a periodic sampling
system with period $1/f_{q}$ and sampling rate $f_{s}$.}
\end{figure*}

In short, a periodic sampling system consists of a periodic preprocessor
followed by a pointwise sampler with a periodic sampling set, as illustrated
in Fig. \ref{fig:PeriodicSamplingGrid}. Since the impulse response
can be arbitrary within a period, this allows us to model multibranch
sampling methods with each branch using the same sampling rate. Periodic
sampling schemes subsume as special cases a broad class of sampling
techniques, e.g. sampling via filter banks, sampling via periodic
modulation, and recurrent nonuniform sampling \cite{Yen1956,EldOpp2000}. 

The periodicity of the sampling system renders the linear operator
associated with the whole system to be block Toeplitz. The asymptotic
spectral properties of block Toeplitz operators (e.g. \cite{Tilli98})
guarantee the existence of $\lim_{T\rightarrow\infty}C_{T}^{\mathcal{P}}\left(f_{s},P\right)$
for a given periodic sampling system $\mathcal{P}$, and allows a
capacity expression to be obtained in terms of the Fourier representation.
Denote by $Q_{k}(f)$ the Fourier transform of the impulse response
$q(t_{k},t_{k}-t)$ of the sampling system, i.e. $Q_{k}(f):=\int_{-\infty}^{\infty}q(t_{k},t_{k}-t)\exp(-j2\pi ft)\mathrm{d}t$.
We further introduce an $f_{s}T_{q}\times\infty$ dimensional Fourier
series matrix $\boldsymbol{F}_{q}\left(f\right)$ associated with
the sampling system, and another infinite diagonal square matrix $\boldsymbol{F}_{h}\left(f\right)$
associated with the channel response. For all $m,l\in\mathbb{Z}$
and $1\leq k\leq f_{s}T_{q}$, we set
\begin{equation}
\begin{cases}
\left(\boldsymbol{F}_{q}\right)_{k,l}\left(f\right) & :=Q_{k}\left(f+lf_{q}\right),\\
\left(\boldsymbol{F}_{h}\right)_{l,l}\left(f\right) & :=\frac{H\left(f+lf_{q}\right)}{\sqrt{\mathcal{S}\left(f+lf_{q}\right)}}.
\end{cases}
\end{equation}
We can then express the sampled analog capacity for a given periodic
system $\mathcal{P}$ in closed form as follows%
\footnote{We note that a periodic sampling system can be equivalently converted
to a filterbank sampling system, and hence \cite[Theorem 4]{ChenGolEld2010}
immediately leads to (\ref{eq:CapacityPeriodicSampling}). While the
analysis framework in \cite{ChenGolEld2010} relies on a discretization
argument, we provide here a more general approach based on operator
analysis that no longer requires discretization. More importantly,
this approach allows to accommodate any sampled channel with integrable
$\left|H(f)Q_{k}(f)\right|^{2}/\mathcal{S}_{\eta}(f)$, which subsumes
channels under\emph{ }band-limited filterbank sampling. In comparison,
\cite{ChenGolEld2010} requires $\mathcal{F}^{-1}\left(\left|H(f)Q_{k}(f)\right|/\sqrt{\mathcal{S}_{\eta}(f)}\right)$
to lie in $\mathcal{L}_{1}$, which is far more restrictive (e.g.
it does not subsume band-limited filterbank sampling as the sinc function
is not absolutely integrable).%
}.

\begin{lem}\label{thm:PeriodicSampledCapacity}Suppose the sampling
system $\mathcal{P}$ is periodic with period $T_{q}$ and sampling
rate $f_{s}$, where $f_{s}T_{q}\in\mathbb{Z}$. Let $f_{q}:=1/T_{q}$.
Assume that $\mathcal{S}_{\eta}(f)\neq0$ for every $f$, $\left|H(f)Q_{k}(f)\right|^{2}/\mathcal{S}_{\eta}(f)$
is bounded and integrable for all $1\leq k\leq f_{s}T_{q}$. We also
assume that the smallest singular value of $\boldsymbol{F}_{q}\boldsymbol{F}_{q}^{*}\left(f\right)$
is uniformly bounded away from 0 over all $f\in\left[-f_{q}/2,f_{q}/2\right]$. 

(1) The sampled channel capacity under optimal power allocation is
given by
\begin{align}
C^{\mathcal{P}}\left(P\right) & =\frac{1}{2}{\displaystyle \int}_{-f_{q}/2}^{f_{q}/2}\sum_{i=1}^{f_{s}T_{q}}\log\left[\nu\cdot\lambda_{i}\right]^{+}\mathrm{d}f,\label{eq:CapacityPeriodicSampling}
\end{align}
where $\nu$ satisfies
\[
{\displaystyle \int}_{-f_{q}/2}^{f_{q}/2}\sum_{i=1}^{f_{s}T_{q}}\left[\nu-\frac{1}{\lambda_{i}}\right]^{+}\mathrm{d}f=P.
\]
Here, $\lambda_{i}$ denotes the $i$th largest eigenvalue of the
matrix $\left(\boldsymbol{F}_{q}\boldsymbol{F}_{q}^{*}\right)^{-\frac{1}{2}}\boldsymbol{F}_{q}\boldsymbol{F}_{h}\boldsymbol{F}_{h}^{*}\boldsymbol{F}_{q}^{*}\left(\boldsymbol{F}_{q}\boldsymbol{F}_{q}^{*}\right)^{-\frac{1}{2}}$.

(2) Suppose further that $H(f)=0$ for any $f\notin[0,W]$, and that
the transmitter employs equal power allocation over $[0,W]$. Then
the sampled channel capacity is given by
\begin{align}
C_{\mathrm{eq}}^{\mathcal{P}}\left(P\right) & ={\displaystyle \int}_{-f_{q}/2}^{f_{q}/2}\frac{1}{2}\log\left(\boldsymbol{I}+\frac{P}{W}\left(\boldsymbol{F}_{q}\boldsymbol{F}_{q}^{*}\right)^{-\frac{1}{2}}\boldsymbol{F}_{q}\right.\nonumber \\
 & \quad\quad\quad\left.\cdot\boldsymbol{F}_{h}\boldsymbol{F}_{h}^{*}\boldsymbol{F}_{q}^{*}\left(\boldsymbol{F}_{q}\boldsymbol{F}_{q}^{*}\right)^{-\frac{1}{2}}\right)\mathrm{d}f.\label{eq:CapacityPeriodicNonWaterFilling}
\end{align}
\end{lem}

\begin{IEEEproof}See Appendix \ref{sec:Proof-of-Theorem-Periodic-Sampled-Capacity}.
\end{IEEEproof}

In Lemma \ref{thm:PeriodicSampledCapacity}, $\nu$ is the water-level
with respect to the optimal water-filling power allocation strategy
over the eigenvalues of the matrix $\left(\boldsymbol{F}_{q}\boldsymbol{F}_{q}^{*}\right)^{-\frac{1}{2}}\boldsymbol{F}_{q}\boldsymbol{F}_{h}\boldsymbol{F}_{h}^{*}\boldsymbol{F}_{q}^{*}\left(\boldsymbol{F}_{q}\boldsymbol{F}_{q}^{*}\right)^{-\frac{1}{2}}$.
The capacity expression (\ref{eq:CapacityPeriodicSampling}) admits
a simple upper bound, as stated below.

\begin{corollary}\label{cor:PeriodicCapacityUpperBound}(a) Consider
the setup and assumptions in Lemma \ref{thm:PeriodicSampledCapacity}.
Under all periodic sampling systems with period $T_{q}$ and sampling
rate $f_{s}$, the sampled channel capacity can be bounded above by
\begin{equation}
C_{f_{q}}\left(f_{s},P\right)=\frac{1}{2}{\displaystyle \int}_{-f_{q}/2}^{f_{q}/2}\sum_{i=1}^{f_{s}T_{q}}\left[\log\left(\nu_{\mathrm{p}}\lambda_{i}\left\{ \boldsymbol{F}_{h}\boldsymbol{F}_{h}^{*}\right\} \right)\right]^{+}\mathrm{d}f,\label{eq:UpperBoundPeriodicSamplingSameRate}
\end{equation}
where $\nu_{\mathrm{p}}$ satisfies
\[
{\displaystyle \int}_{-f_{q}/2}^{f_{q}/2}\sum_{i=1}^{f_{s}T_{q}}\left[\nu_{\mathrm{p}}-\frac{1}{\lambda_{i}\left\{ \boldsymbol{F}_{h}\boldsymbol{F}_{h}^{*}\right\} }\right]^{+}\mathrm{d}f=P.
\]

(b) Suppose that there exists a frequency set $B_{\mathrm{m}}$ that
satisfies $\mu\left(B_{\mathrm{m}}\right)=f_{s}$ and
\[
{\displaystyle \int}_{f\in B_{\mathrm{m}}}\frac{\left|H(f)\right|^{2}}{\mathcal{S}_{\eta}(f)}\mathrm{d}f=\sup_{B:\mu\left(B\right)=f_{s}}{\displaystyle \int}_{f\in B}\frac{\left|H(f)\right|^{2}}{\mathcal{S}_{\eta}(f)}\mathrm{d}f.
\]
Then
\begin{equation}
C_{f_{q}}\left(f_{s},P\right)\leq C_{\mathrm{u}}\left(f_{s},P\right),\label{eq:CorGeneralUpperBound}
\end{equation}
where $C_{\mathrm{u}}\left(f_{s},P\right)$ is given by (\ref{eq:GeneralCapacityUpperBound}).

{} \end{corollary}

\begin{IEEEproof}(a) Following the same steps as in \cite[Proposition 1]{ChenGolEld2010},
we can see that the $i$th largest eigenvalue satisfies
\[
\lambda_{i}\left\{ \left(\boldsymbol{F}_{q}\boldsymbol{F}_{q}^{*}\right)^{-\frac{1}{2}}\boldsymbol{F}_{q}\boldsymbol{F}_{h}\boldsymbol{F}_{h}^{*}\boldsymbol{F}_{q}^{*}\left(\boldsymbol{F}_{q}\boldsymbol{F}_{q}^{*}\right)^{-\frac{1}{2}}\right\} \leq\lambda_{i}\left(\boldsymbol{F}_{h}\boldsymbol{F}_{h}^{*}\right),
\]
which immediately leads to (\ref{eq:UpperBoundPeriodicSamplingSameRate}). 

(b) For any given $f_{q}$, the upper bound (\ref{eq:UpperBoundPeriodicSamplingSameRate})
is obtained by extracting out a certain frequency set $B$ that has
measure $\mu(B)=f_{s}$ and suppressing all spectral components outside
$B$. By our definition of $B_{m}$, any choice of $B$ with spectral
size $f_{s}$ will not outperform $B_{m}$. Hence, choosing $B=B_{m}$
leads to a universal upper bound.\end{IEEEproof}

Corollary \ref{cor:PeriodicCapacityUpperBound} reveals that the capacity
under any periodic sampling system, no matter what its period is,
cannot exceed the upper bound $C_{\mathrm{u}}\left(f_{s},P\right)$
in Theorem \ref{thm:GeneralSampledCapacity}. Our remaining proof
is then established by observing that any aperiodic sampling system
can be related to a periodic sampling system by truncation and periodization,
as elaborated in the next two subsections. 

\subsection{General Upper Bound: Finite-duration $h(t)$}

In this subsection, we focus on the channel whose impulse response
is of finite duration $2L_{0}$, i.e. 
\[
h(t)=0,\quad\forall t\text{ }(|t|>L_{0}).
\]
Our goal is to prove that the capacity upper bound (\ref{eq:GeneralCapacityUpperBound})
holds for this type of channel. 

For any transmission block of duration $2T$, we call the transmit
signal $x(t)$ over this block a codeword (or symbol) of code length
$2T$. The information conveyed through such finite-duration codewords
can be bounded via certain analog channel capacity, as long as we
can preclude inter-symbol interference. The key idea here is to separate
consecutive codewords with a guard zone with sufficient length and
then use capacity-achieving strategies separately for each codeword.
When the code length $2T$ is sufficiently large, the transmission
time wasted on the guard zones becomes negligible, which in turn allows
us to approach the true capacity arbitrarily well. The detailed analysis
proceeds as follows. 

\textbf{Step 1}. Consider an input $x(t)$ that is constrained to
the interval $[-T,T]$. Since $h(t)$ is of finite duration $2L_{0}$,
the channel output $r(t)=h(t)*x(t)+\eta(t)$ will be affected by the
input only when $t\in[-T-L_{0},T+L_{0}]$. Define a \emph{window operator}
and its complement operator such that
\begin{equation}
w_{T}(f(t))=\begin{cases}
f(t),\quad & \text{if }\left|t\right|\leq T+L_{0},\\
0, & \text{else};
\end{cases}
\end{equation}
and
\begin{equation}
w_{T}^{\perp}(f(t))=\begin{cases}
0,\quad & \text{if }\left|t\right|\leq T+L_{0},\\
f(t), & \text{else}.
\end{cases}
\end{equation}
Then for any linear sampling operator $\mathcal{P}$ with impulse
response $q(t,\tau)$, the sampled output is $\mathcal{P}\left(r(t)\right)=\mathcal{P}\left(w_{T}\left(r(t)\right)\right)+\mathcal{P}\left(w_{T}^{\perp}\left(r(t)\right)\right)$.
One can easily observe that the component $\mathcal{P}\left(w_{T}^{\perp}\left(r(t)\right)\right)$
contains no information about $x(t)$, and is statistically \emph{independent}
of $\mathcal{P}\left(w_{T}\left(r(t)\right)\right)$ due to the whiteness
assumption of the noise. In other words, the sampling input outside
the interval $[-T-L_{0},T+L_{0}]$ does not improve capacity at all.
Consequently, it suffices to restrict attention to the class of sampling
systems whose system input is constrained to the interval $[-T-L_{0},T+L_{0}]$.

\textbf{Step 2}. Construct a periodization of the above sampled channel
model with finite input duration. Set the impulse response $q_{T+L_{0}}^{\text{p}}(t,\tau)$
of the preprocessor of the periodized sampling system to be a periodic
extension of $q(t,\tau)$ in the block $[-T-L_{0},T+L_{0}]\times[-T,T]$.
Specifically, if $\tau=k\cdot2\left(T+L_{0}\right)+\tau_{\text{r}}$
for some $k\in\mathbb{Z}$ and $\tau_{\text{r}}\in[-T-L_{0},T+L_{0}]$,
then
\begin{equation}
q_{T+L_{0}}^{\text{p}}(t,\tau)=\begin{cases}
q\left(t-2k(T+L_{0}),\tau_{\text{r}}\right), & \text{if }\left|t-2k(T+L_{0})\right|\\
 & \leq T+L_{0},\\
0, & \text{else.}
\end{cases}\label{eq:ImpulseResponsePeriodizedSystemFiniteDuration}
\end{equation}
Apparently, $q_{T+L_{0}}^{\text{p}}(t,\tau)$ corresponds to a periodic
preprocessing system with period $2\left(T+L_{0}\right)$.

Suppose without loss of generality that the indices of the sample
times that fall in $[-T-L_{0},T+L_{0}]$ are $0,1,\cdots,K-1$, i.e.
$\left\{ k\mid t_{k}\in[-T-L_{0},T+L_{0}]\right\} =\left\{ 0,1,\cdots,K-1\right\} $.
We can then set the sampling set $\Lambda_{T+L_{0}}^{\text{p}}$ of
the periodized system such that for any sampling time $t_{k}\in\Lambda_{T+L_{0}}^{\text{p}}$,
we have
\begin{equation}
t_{k}=t_{k\text{ mod }K}+2(T+L_{0})\cdot\left\lfloor \frac{k}{K}\right\rfloor ,\label{eq:SamplingSetPeriodizedSystemFiniteDuration}
\end{equation}
where $\left\lfloor x\right\rfloor \overset{\Delta}{=}\max\left\{ n\mid n\in\mathbb{Z},n\leq x\right\} $.
Clearly, this forms a periodic sampling set with period $2\left(T+L_{0}\right)$.
The definition of Beurling density ensures that for any $\epsilon>0$,
there exists a $T_{D}$ such that for every $T>T_{D}$,
\[
f_{s}-\epsilon\leq D\left(\Lambda_{T+L_{0}}^{\text{p}}\right)\leq f_{s}+\epsilon.
\]

Due to the finite-duration assumption of $h(t)$, our construction
(\ref{eq:ImpulseResponsePeriodizedSystemFiniteDuration}) guarantees
that the input $x(t)$ within time interval $[2k(T+L_{0})-T,2k(T+L_{0})+T]$
will only affect the sampled output at the \emph{$k$th time block}
$[\left(2k-1\right)(T+L_{0}),\left(2k+1\right)(T+L_{0})]$, as illustrated
in Fig. \ref{fig:GuardZone}. Since the noise $\eta(t)$ is assumed
to be white, the noise components across different time blocks are
independent. In fact, the intervals $[2k(T+L_{0})+T,(2k+1)(T+L_{0})-T]$
($k\in\mathbb{Z}$) act effectively as \emph{guard zones} in order
to avoid leakage of signals across different time blocks.

\begin{figure}[htbp]
\begin{centering}
\textsf{\includegraphics[scale=0.43]{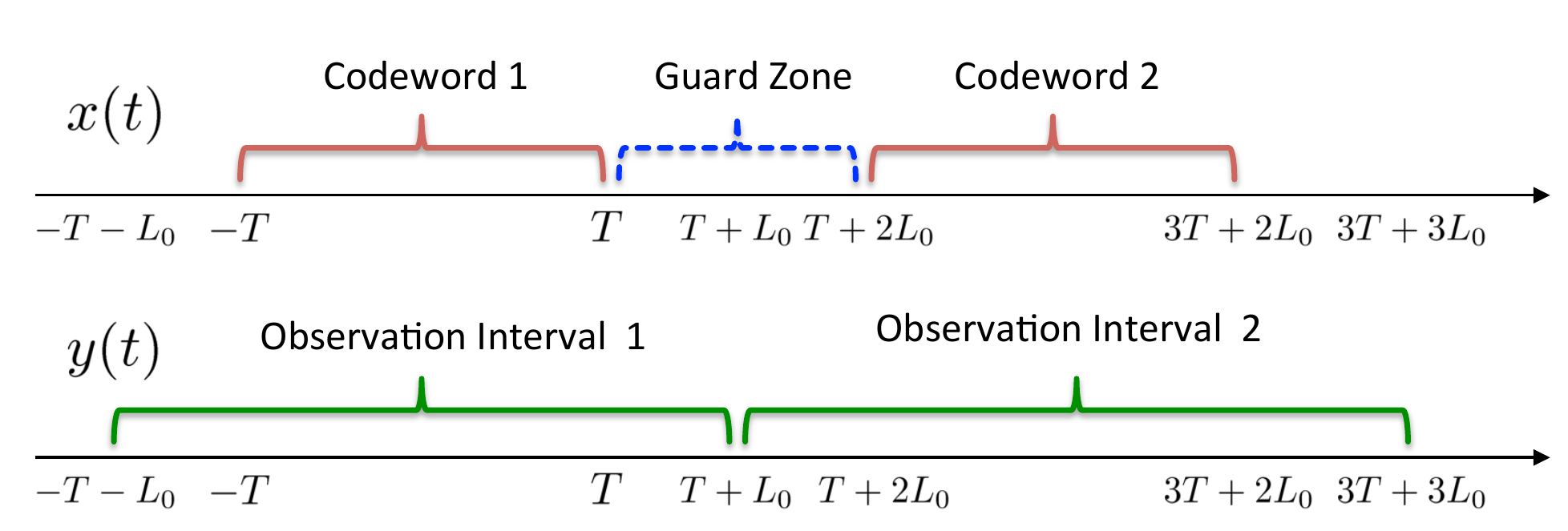}} 
\par\end{centering}

\caption{\label{fig:GuardZone} The code words of duration $2T$ are separated
by guard zones of duration $2L_{0}$. There is no inter-symbol interference
among different observation intervals.}
\end{figure}

Based on the above argument, we can separate codewords of duration
$2T$ in $[2k(T+L_{0})+T,(2k+1)(T+L_{0})-T]$ $\left(k\in\mathbb{Z}\right)$
on the analog channel by a guard zone $2L_{0}$ (as illustrated in
Fig. \ref{fig:GuardZone}). The ratio of guard space to the length
of the time block vanishes as $T\rightarrow\infty$, and there is
no intersymbol interference under the new system we construct. By
our capacity definition, for any $\delta>0$, there exists a $T_{0}$
such that $\forall T>T_{0}$, we have 
\[
\frac{T+L_{0}}{T}<1+\delta,\quad\text{and}\quad\frac{T}{T+L_{0}}>1-\delta.
\]
Consequently,
\begin{align}
C_{T}^{\mathcal{P}}\left(P\right) & \overset{(\text{i})}{\leq}\frac{T+L_{0}}{T}C_{\text{p}}^{\mathcal{P}}\left(\frac{T}{T+L_{0}}P\right)\nonumber \\
 & \leq\left(1+\delta\right)C_{\text{p}}^{\mathcal{P}}\left(\left(1-\delta\right)P\right),\label{eq:CpCLambda_UB}
\end{align}
where $C_{\text{p}}^{\mathcal{P}}$ denotes the capacity under our
periodized sampling system. The inequality (i) follows from the following
three arguments: 
\begin{itemize}
\item $C_{T}^{\mathcal{P}}$ is the information rate when we observe the
samples within the interval $\left[-T,T\right]$, which is smaller
than the information rate, termed $\hat{C}_{T}^{\mathcal{P}}$, when
we observe all samples within $\left[-T-L_{0},T+L_{0}\right]$; 
\item $\hat{C}_{T}^{\mathcal{P}}$ is equivalent to the maximum information
rate achievable by the periodized system, under the constraint that
there is no input signal transmitted over the guard zones. Clearly,
this rate will be smaller than the capacity without this transmission
constraint, which is $\frac{T+L_{0}}{T}C_{\text{p}}^{\mathcal{P}}$.
Here, the multiplication factor $\frac{T+L_{0}}{T}$ arises from the
fact that we only use a portion $\frac{T}{T+L_{0}}$ of time for transmission; 
\item Since the total energy over each transmission block is $PT$ and each
guard zone has zero power, the average power allocated to the transmitted
signal is $\frac{T}{T+L_{0}}P$.
\end{itemize}
We know from Corollary \ref{cor:PeriodicCapacityUpperBound}(b) that
\begin{align}
C_{\text{p}}^{\mathcal{P}}\left(\left(1-\delta\right)P\right) & \leq C_{\text{u}}\left(D\left(\Lambda_{T+L_{0}}^{\text{p}}\right),\left(1-\delta\right)P\right)\nonumber \\
 & \leq C_{\text{u}}\left(f_{s}+\epsilon,\left(1-\delta\right)P\right),\label{eq:CpCu_UB}
\end{align}
where the last inequality arises from observing that $C_{\text{u}}(f_{s},P)$
is monotonically non-decreasing in $f_{s}$ and $P$. Putting (\ref{eq:CpCLambda_UB})
and (\ref{eq:CpCu_UB}) together yields
\begin{align}
C_{T}^{\mathcal{P}}\left(P\right) & \leq\left(1+\delta\right)C_{\text{u}}\left(f_{s}+\epsilon,P\right)\label{eq:FiniteDurationChannelInequality}
\end{align}
as soon as $T>\max\left\{ T_{0},T_{D}\right\} $. Since $\epsilon$
and $\delta$ can be chosen arbitrarily small, we have that
\[
\lim\sup_{T\rightarrow\infty}C_{T}^{\mathcal{P}}(P)\leq C_{\text{u}}\left(f_{s},P\right)
\]
when $h(t)$ is of finite duration and $\eta(t)$ is white.

\subsection{General Upper Bound: Infinite-Duration $h(t)$}

We now investigate the capacity bound when $h(t)$ is not time-limited.
We would like to prove that for any given sampling system $\mathcal{P}$
and any $\epsilon>0$, there exists $T_{1}$ such that for any $T>T_{1}$
, one has
\[
C_{T}^{\mathcal{P}}\leq C_{\text{u}}\left(f_{s},P\right)+\epsilon.
\]

Our proof proceeds by comparing the original channel with a truncated
channel whose channel response $\tilde{h}(f)$ satisfies
\[
\tilde{h}(t)\overset{\Delta}{=}\begin{cases}
h(t),\quad & \text{if }\left|t\right|\leq L_{1},\\
0, & \text{otherwise}.
\end{cases}
\]
Let $\xi>0$ be an arbitrary small constant, and $L_{1}$ chosen such
that
\begin{equation}
\int_{-\infty}^{-L_{1}}\left|h(t)\right|^{2}\mathrm{d}t+\int_{L_{1}}^{\infty}\left|h(t)\right|^{2}\mathrm{d}t\leq\xi.\label{eq:ChoiceL1TruncatedChannel}
\end{equation}
We further constrain the input and the observed sampled output to
the time interval $[-T,T]$. For both the original and truncated channel,
the sampled noise is not white, which motivates us to first perform
prewhitening.

Suppose without loss of generality that the sampled times within $[-T,T]$
are $\{t_{i}\mid1\leq i\leq K_{T}\}$. For convenience of notation,
we introduce a linear operator $\hat{\mathcal{P}}_{T}$ associated
with the sampling system such that
\[
\hat{\mathcal{P}}_{T}(\hat{r}(t))=\left[y[1],y[2],\cdots,y[K_{T}]\right],
\]
where $\hat{r}(t)=g(x)*x(t)+\hat{\eta}(t)$ is the sampling system
input, $\hat{\eta}(t)$ is white, and $\left\{ y[n]\right\} $ are
the corresponding sampled output. Thus, for the original channel,
one can write
\[
\left[\begin{array}{c}
y[1]\\
y[2]\\
\vdots\\
y[K_{T}]
\end{array}\right]=\hat{\mathcal{P}}_{T}\left(g(t)*x(t)\right)+\hat{\mathcal{P}}_{T}\left(\hat{\eta}(t)\right).
\]

Denote by $\hat{q}(t_{i},\tau)$ the impulse response associated with
this sampling system. Then, the noise component $\hat{\mathcal{P}}_{T}\left(\hat{\eta}(t)\right)$
can be whitened by left-multiplying it with a $K_{T}$-dimensional
square matrix $\boldsymbol{W}_{\hat{\mathcal{P}}}^{-1/2}$ defined
by
\[
\boldsymbol{W}_{\hat{\mathcal{P}}}(i,j)=\int_{-\infty}^{\infty}\hat{q}\left(t_{i},\tau\right)\hat{q}^{*}\left(t_{j},\tau\right)\mathrm{d}\tau.
\]
The invertibility is guaranteed by our assumptions. To see this, if
we denote by $\tilde{\eta}\overset{\Delta}{=}\boldsymbol{W}_{\hat{\mathcal{P}}}^{-1/2}\hat{\mathcal{P}}_{T}\left(\hat{\eta}(t)\right)$
the $K_{T}$-dimensional ``prewhitened'' noise, then one can verify
that for every $i$ and $j$,
\begin{align*}
 & \text{ }\left[\mathbb{E}\left(\hat{\mathcal{P}}_{T}\left(\hat{\eta}(t)\right)\left(\hat{\mathcal{P}}_{T}\left(\hat{\eta}(t)\right)\right)^{\top}\right)\right]_{ij}\\
 & \text{ }\text{ }=\mathbb{E}\left[\left(\int_{-\infty}^{\infty}\hat{q}\left(t_{i},\tau\right)\hat{\eta}\left(\tau\right)\mathrm{d}\tau\right)\left(\int_{-\infty}^{\infty}\hat{q}^{*}\left(t_{j},\tau\right)\hat{\eta}(\tau)\mathrm{d}\tau\right)\right]\\
 & \text{ }\text{ }=\int_{-\infty}^{\infty}\hat{q}\left(t_{i},\tau\right)\hat{q}^{*}\left(t_{j},\tau\right)\mathrm{d}\tau\\
 & \text{ }\text{ }=\boldsymbol{W}_{\hat{\mathcal{P}}}(i,j)
\end{align*}
or, equivalently, 
\[
\mathbb{E}\left[\hat{\mathcal{P}}_{T}\left(\hat{\eta}(t)\right)\left(\hat{\mathcal{P}}_{T}\left(\hat{\eta}(t)\right)\right)^{\top}\right]=\boldsymbol{W}_{\hat{\mathcal{P}}}.
\]
As a result, the covariance of $\tilde{\eta}$ obeys
\begin{align*}
\mathbb{E}\left[\tilde{\eta}\tilde{\eta}^{\top}\right] & =\boldsymbol{W}_{\hat{\mathcal{P}}}^{-1/2}\mathbb{E}\left(\hat{\mathcal{P}}_{T}\left(\hat{\eta}(t)\right)\left(\hat{\mathcal{P}}_{T}\left(\hat{\eta}(t)\right)\right)^{\top}\right)\boldsymbol{W}_{\hat{\mathcal{P}}}^{-1/2}\\
 & =\boldsymbol{I}.
\end{align*}
If we denote by $\hat{\mathcal{P}}_{\text{w}}\overset{\Delta}{=}\boldsymbol{W}_{\hat{\mathcal{P}}}^{-\frac{1}{2}}\cdot\hat{\mathcal{P}}_{T}$
and let $\hat{q}_{\text{w}}\left(t_{i},\tau\right)$ represent its
associated impulse response, then the above calculation reveals that
\begin{equation}
\int_{-\infty}^{\infty}\hat{q}_{\text{w}}\left(t_{i},\tau\right)\hat{q}_{\text{w}}^{*}\left(t_{j},\tau\right)\mathrm{d}\tau=\begin{cases}
1,\quad & \text{if }i=j;\\
0, & \text{else},
\end{cases}\label{eq:PropertyProjectionOperator}
\end{equation}
indicating that $\left\{ \hat{q}_{\text{w}}\left(t_{i},\cdot\right),1\leq i\leq K_{T}\right\} $
forms a set of orthonormal sequences in the corresponding Hilbert
space. 

For an operator $\mathcal{A}$ with an impulse response $a(t,\tau)$
($-T\leq\tau\leq T$, $t\in\{t_{i}\mid1\leq i\leq K_{T}\}$) and input
domain $\mathcal{D}\left(\mathcal{A}\right)$, we denote by $\left\Vert \mathcal{A}\right\Vert _{\text{F}}$
the generalized Frobenius norm of the operator $\mathcal{A}$ with
respect to its associated domain, namely,
\[
\left\Vert \mathcal{A}\right\Vert _{\text{F}}:=\begin{cases}
\sqrt{\sum_{i=1}^{K_{T}}\int_{-T}^{T}\left|a(t_{i},\tau)\right|^{2}\mathrm{d}\tau},\\
\quad\quad\quad\text{if }\mathcal{D}\left(\mathcal{A}\right)=\{t_{i}\mid1\leq i\leq K_{T}\}\times[-T,T],\\
\sqrt{\sum_{i=1}^{K_{T}}\int_{-\infty}^{\infty}\left|a(t_{i},\tau)\right|^{2}\mathrm{d}\tau},\\
\quad\quad\quad\text{if }\mathcal{D}\left(\mathcal{A}\right)=\{t_{i}\mid1\leq i\leq K_{T}\}\times[-\infty,\infty],\\
\sqrt{\int_{-\infty}^{\infty}\int_{-T}^{T}\left|a(t,\tau)\right|^{2}\mathrm{d}\tau\mathrm{d}t},\\
\quad\quad\quad\text{if }\mathcal{D}\left(\mathcal{A}\right)=[-\infty,\infty]\times[-T,T].
\end{cases}
\]
Recall that $\hat{q}_{\text{w}}\left(t_{i},\cdot\right)$ ($1\leq i\leq K_{T}$)
forms orthonormal sequences. By Bessel's inequality \cite{Kreyszig1989},
an operator $\mathcal{A}$ with $\mathcal{D}\left(\mathcal{A}\right)=[-\infty,\infty]\text{ }\times\text{ }[-T,T]$
satisfies
\[
\sum_{i=1}^{K_{T}}\left|\left\langle \hat{q}_{\text{w}}\left(t_{i},\cdot\right),a\left(\cdot,\tau\right)\right\rangle \right|^{2}\leq\int_{-\infty}^{\infty}\left|a\left(\tau_{1},\tau\right)\right|^{2}\mathrm{d}\tau_{1}
\]
for every $\tau\in\left[-T,T\right]$, which immediately gives
\begin{align*}
\left\Vert \hat{\mathcal{P}}_{\text{w}}\mathcal{A}\right\Vert _{\text{F}}^{2} & =\sum_{i=1}^{K_{T}}\int_{-T}^{T}\left|\int_{-\infty}^{\infty}\hat{q}_{\text{w}}(t_{i},\tau_{1})a(\tau_{1},\tau)\mathrm{d}\tau_{1}\right|^{2}\mathrm{d}\tau\\
 & =\int_{-T}^{T}\sum_{i=1}^{K_{T}}\left|\left\langle \hat{q}_{\text{w}}\left(t_{i},\cdot\right),a\left(\cdot,\tau\right)\right\rangle \right|^{2}\mathrm{d}\tau\\
 & \leq\int_{-T}^{T}\int_{-\infty}^{\infty}\left|a\left(\tau_{1},\tau\right)\right|^{2}\mathrm{d}\tau_{1}\mathrm{d}\tau\leq\left\Vert \mathcal{A}\right\Vert _{\text{F}}^{2}.
\end{align*}

Denote by $\left\{ \lambda_{i}\right\} $ and $\{\tilde{\lambda}_{i}\}$
the set of\emph{ squared} singular values associated with the original
sampled channel operator $\hat{\mathcal{P}}_{\text{w}}\mathcal{G}$
and the operator $\hat{\mathcal{P}}_{\text{w}}\tilde{\mathcal{G}}$
of the truncated sampled channel, respectively. Here, $\mathcal{G}$
and $\tilde{\mathcal{G}}$ represent respectively the operator associated
with the original channel response and the truncated channel response.
We can obtain some properties of $\left\{ \lambda_{i}\right\} $ and
$\{\tilde{\lambda}_{i}\}$ as stated in the following lemma.

\begin{lem}\label{lem:PropertiesLambda}Suppose that $\int_{-\infty}^{\infty}\left|g(t)\right|^{2}\mathrm{d}t<C_{g}<\infty$
for some constant $C_{g}$. For any $\xi>0$, there exists $T_{0}$
such that for every $T>T_{0}$, one has

(1) $\left|\frac{1}{2T}\sum_{i}\lambda_{i}-\frac{1}{2T}\sum_{i}\tilde{\lambda}_{i}\right|\leq\xi+2\sqrt{\xi C_{g}}.$

(2) $\frac{1}{2T}\sum_{i}\lambda_{i}\leq\int_{-\infty}^{\infty}\left|g(t)\right|^{2}\mathrm{d}t<\infty$.

(3) Suppose that $h(t)=O\left(\frac{1}{t^{1.5+\varepsilon}}\right)$
for some small $\varepsilon>0$. Then there exists $T_{0,\epsilon}$
such that for every $T>T_{0,\epsilon}$, one has $\left|\lambda_{i}-\tilde{\lambda}_{i}\right|\leq\xi$.\end{lem}

\begin{IEEEproof}See Appendix \ref{sec:Proof-of-Lemma-PropertiesLambda}.\end{IEEEproof}

For notational simplicity, define two functions as follows
\begin{align}
C_{T}^{\mathcal{P}}\left(\nu,\left\{ \lambda_{i}\right\} \right) & :=\frac{1}{2T}\sum_{i=1}^{K_{T}}\frac{1}{2}\left[\log\left(\nu\lambda_{i}\right)\right]^{+}
\end{align}
and
\begin{equation}
F_{T}\left(\nu,\left\{ \lambda_{i}\right\} \right):=\frac{1}{2T}\sum_{i=1}^{K_{T}}\left[\nu-\frac{1}{\lambda_{i}}\right]^{+}
\end{equation}
for some water level $\nu$. Note that if $\nu$ is chosen such that
$F_{T}\left(\nu,\left\{ \lambda_{i}\right\} \right)=P$, then
\[
C_{T}^{\mathcal{P}}\left(\nu,\left\{ \lambda_{i}\right\} \right)=C_{T}^{\mathcal{P}}\left(P\right).
\]
Apparently, both $C_{T}^{\mathcal{P}}\left(P\right)$ and $C_{T}^{\mathcal{P}}\left(\nu,\left\{ \lambda_{i}\right\} \right)$
are non-decreasing functions of $\left\{ \lambda_{i}\right\} $, which
implies that
\begin{align}
C_{T}^{\mathcal{P}}\left(\nu,\left\{ \lambda_{i}\right\} \right) & \leq C_{T}^{\mathcal{P}}\left(\nu,\left\{ \max\left\{ \lambda_{i},\xi^{\frac{1}{3}}\right\} \right\} \right)
\end{align}
and 
\[
C_{T}^{\mathcal{P}}\left(P\right)\leq C_{T}^{\mathcal{P}}\left(\nu,\left\{ \max\left\{ \lambda_{i},\xi^{\frac{1}{3}}\right\} \right\} \right),
\]
where $\nu$ is determined by 
\begin{equation}
F_{T}\left(\nu,\left\{ \max\left(\lambda_{i},\xi^{\frac{1}{3}}\right)\right\} \right)=P.\label{eq:PowerFt}
\end{equation}
Here, $\xi>0$ is some arbitrarily small constant. In fact, one can
easily verify that $C_{T}^{\mathcal{P}}\left(\nu,\left\{ \max\left\{ \lambda_{i},\xi^{\frac{1}{3}}\right\} \right\} \right)$
with $\nu$ determined by (\ref{eq:PowerFt}) is no larger than the
sum capacity of two separate channels with respective eigenvalues
$\left\{ \lambda_{i}\right\} $ and $\left\{ \breve{\lambda}_{i}:=\xi^{\frac{1}{3}}\right\} $
each with power allocation $P$. In other words, 
\begin{align}
 & C_{T}^{\mathcal{P}}\left(\nu,\left\{ \max\left\{ \lambda_{i},\xi^{\frac{1}{3}}\right\} \right\} \right)\nonumber \\
 & \text{ }\text{ }\leq C_{T}^{\mathcal{P}}\left(P\right)+C_{T}^{\mathcal{P}}\left(\nu,\left\{ \xi^{\frac{1}{3}}\right\} _{1\leq i\leq K_{T}}\right)\\
 & \text{ }\text{ }\leq C_{T}^{\mathcal{P}}\left(P\right)+\frac{K_{T}}{2T}\log\left(1+\frac{PT}{K_{T}}\xi^{\frac{1}{3}}\right)\nonumber \\
 & \text{ }\text{ }\leq C_{T}^{\mathcal{P}}\left(P\right)+\frac{K_{T}}{2T}\cdot\frac{PT}{K_{T}}\xi^{\frac{1}{3}}\nonumber \\
 & \text{ }\text{ }=C_{T}^{\mathcal{P}}\left(P\right)+\frac{P}{2}\xi^{\frac{1}{3}}.\label{eq:CtUB_Ctmax}
\end{align}

For any positive water level $\nu$ and some small constant $\xi>0$,
the Lipschitz constants of the functions
\begin{align*}
f_{1}\left(x\right) & :=\frac{1}{2}\left[\log\left(\nu\max\left\{ x,\xi^{\frac{1}{3}}\right\} \right)\right]^{+}\\
f_{2}\left(x\right) & :=\left[\nu-\max\left\{ x,\xi^{\frac{1}{3}}\right\} ^{-1}\right]^{+}
\end{align*}
 are bounded above in magnitude by $\frac{1}{2}\xi^{-1/3}$ and $\xi^{-2/3}$,
respectively. Using the same water level $\nu$, the corresponding
power for both channels can be computed as
\begin{align*}
P & =\frac{1}{2T}\sum_{i=1}^{K_{T}}\left[\nu-\frac{1}{\max\left\{ \lambda_{i},\xi^{\frac{1}{3}}\right\} }\right]^{+},\\
\tilde{P} & =\frac{1}{2T}\sum_{i=1}^{K_{T}}\left[\nu-\frac{1}{\max\left\{ \tilde{\lambda}_{i},\xi^{\frac{1}{3}}\right\} }\right]^{+}.
\end{align*}
Combining Lemma \ref{lem:PropertiesLambda} and the Lipschitz constants
of $f_{2}\left(x\right)$ immediately suggests that: there exists
$T_{0,\epsilon}$ such that for any $T>T_{0,\epsilon}$, one has
\begin{align}
\left|\tilde{P}-P\right| & =\frac{1}{2T}\sum_{i=1}^{K_{T}}\frac{1}{\xi^{\frac{2}{3}}}\left|\lambda_{i}-\tilde{\lambda}_{i}\right|\leq\frac{K_{T}}{2T\xi^{\frac{2}{3}}}\xi\nonumber \\
 & \leq\left(f_{s}+\epsilon\right)\xi^{\frac{1}{3}}.\label{eq:differenceP}
\end{align}
Similarly, we can bound
\begin{align}
 & \left|C_{T}^{\mathcal{P}}\left(\nu,\left\{ \max\left\{ \lambda_{i},\xi^{\frac{1}{3}}\right\} \right\} \right)-C_{T}^{\mathcal{P}}\left(\nu,\left\{ \max\left\{ \lambda_{i},\xi^{\frac{1}{3}}\right\} \right\} \right)\right|\nonumber \\
 & \text{ }\text{ }\leq\frac{1}{2T}\sum_{i=1}^{K_{T}}\frac{1}{2\xi^{\frac{1}{3}}}\left|\lambda_{i}-\tilde{\lambda}_{i}\right|\nonumber \\
 & \text{ }\text{ }\leq\frac{1}{4}\left(f_{s}+\epsilon\right)\xi^{\frac{2}{3}}.\label{eq:differenceC_p}
\end{align}

Combining (\ref{eq:differenceP}), (\ref{eq:differenceC_p}) and (\ref{eq:FiniteDurationChannelInequality})
suggests that 
\begin{align}
C_{T}^{\mathcal{P}}\left(P\right)\leq & \text{ }C_{T}^{\mathcal{P}}\left(\nu,\left\{ \max\left\{ \lambda_{i},\xi^{\frac{1}{3}}\right\} \right\} \right)\nonumber \\
\leq & \text{ }C_{T}^{\tilde{\mathcal{P}}}\left(\nu,\left\{ \max\left\{ \tilde{\lambda}_{i},\xi^{\frac{1}{3}}\right\} \right\} \right)+\frac{1}{4}\left(f_{s}+\epsilon\right)\xi^{\frac{2}{3}}\nonumber \\
\leq & C_{T}^{\tilde{\mathcal{P}}}\left(\tilde{P}\right)+\frac{\tilde{P}}{2}\xi^{\frac{1}{3}}+\frac{1}{4}\left(f_{s}+\epsilon\right)\xi^{\frac{2}{3}}\label{eq:intermediate}\\
\leq & \text{ }C_{T}^{\tilde{\mathcal{P}}}\left(P+\left(f_{s}+\epsilon\right)\xi^{\frac{1}{3}}\right)+\frac{\tilde{P}}{2}\xi^{\frac{1}{3}}+\frac{1}{4}\left(f_{s}+\epsilon\right)\xi^{\frac{2}{3}}\nonumber \\
\leq & \frac{P+\left(f_{s}+\epsilon\right)\xi^{\frac{1}{3}}}{2}\xi^{\frac{1}{3}}+\text{ }\frac{1}{4}\left(f_{s}+\epsilon\right)\xi^{\frac{2}{3}}+\left(1+\delta\right)\nonumber \\
 & \quad\cdot C_{\text{u}}\left(f_{s}+\epsilon,P+\left(f_{s}+\epsilon\right)\xi^{\frac{1}{3}}\right),\nonumber 
\end{align}
where (\ref{eq:intermediate}) is a consequence of (\ref{eq:CtUB_Ctmax}).
Since $\delta,\epsilon,$ and $\xi$ can all be made arbitrarily small,
it follows that
\[
\lim\sup_{T\rightarrow\infty}C_{T}^{\mathcal{P}}\left(P\right)\leq C_{\text{u}}\left(f_{s},P\right),
\]
completing the proof.

\section{Proof of Lemma \ref{thm:PeriodicSampledCapacity}\label{sec:Proof-of-Theorem-Periodic-Sampled-Capacity}}

The proof is restricted to the channel with white noise, i.e. $\mathcal{S}_{\eta}(f)\equiv1$.
It is straightforward to extend the analysis to colored noise through
the argument presented in the first paragraph of Appendix \ref{sec:Proof-Architecture-of-Theorem-General-Capacity}.

Our proof proceeds in the following three steps. 
\begin{enumerate}
\item We first introduce several correlation functions and compute the Fourier
series associated with them. These quantities are crucial in deriving
the capacity expression. In particular, when the sampling system is
periodic, the infinite correlation matrices are block Toeplitz.
\item When constrained to a finite time interval $\left[-nT_{q},nT_{q}\right]$,
the sampled output is a finite vector. The sampled noise is in general
not white, which motivates us to whiten it first. In fact, the covariance
matrix of the sampled noise can be easily derived in terms of the
proposed correlation functions.
\item For any time interval $\left[-nT_{q},nT_{q}\right]$, the capacity
is obtained through the Karhunen Loeve expansion. Specifically,
the capacity depends on the eigenvalues of the associated system operator,
which is related to the correlation functions. The asymptotic properties
of block Toeplitz matrices guarantee the convergence when $n\rightarrow\infty$,
which allow us to derive in closed form the sampled channel capacity.
\end{enumerate}

\subsection{Correlation functions and Fourier series}

For a concatenated linear system consisting of the channel filter
followed by the sampling system, we denote by 
\begin{equation}
s\left(t_{\text{o}},t_{\text{i}}\right):=\int_{-\infty}^{\infty}h\left(\tau-t_{\text{i}}\right)q\left(t_{\text{o}},\tau\right)\mathrm{d}\tau
\end{equation}
its system output seen at time $t_{\text{o}}$ due to an impulse input
at time $t_{\text{i}}$. For notational convenience, we define $q_{k}\left(\tau\right):=q\left(t_{k},\tau\right)$
as the sampling output response at time $t_{k}$ due to an impulse
input to the sampling system at time $\tau$. Two \emph{output autocorrelation
functions }are defined as follows
\begin{equation}
\mathcal{R}_{hq}\left(t_{k},t_{l}\right)\overset{\Delta}{=}\int_{-\infty}^{\infty}s\left(t_{k},\tau\right)s^{*}\left(t_{l},\tau\right)\mathrm{d}\tau
\end{equation}
and
\begin{equation}
\mathcal{R}_{q}\left(t_{k},t_{l}\right)\overset{\Delta}{=}\int_{-\infty}^{\infty}q\left(t_{k},\tau\right)q^{*}\left(t_{l},\tau\right)\mathrm{d}\tau.
\end{equation}
For notational simplicity, we use $\mathcal{R}_{hq}(k,l)$ (resp.
$\mathcal{R}_{q}\left(k,l\right)$) and $\mathcal{R}_{hq}\left(t_{k},t_{l}\right)$
(resp. $\mathcal{R}_{q}\left(t_{k},t_{l}\right)$) interchangeably.
When the sampling system is periodic with period $T_{q}$, one can
easily see that both $\left[\mathcal{R}_{hq}(k,l)\right]_{k,l=-\infty}^{\infty}$
and $\left[\mathcal{R}_{q}\left(k,l\right)\right]_{k,l=-\infty}^{\infty}$
are infinite block Toeplitz matrices. 

The spectral properties associated with the system operators are captured
by \emph{Fourier series matrices} $\boldsymbol{F}_{hq}$, $\boldsymbol{F}_{qq}$,
$\boldsymbol{F}_{h}$ and $\boldsymbol{F}_{q}$. Specifically, $\boldsymbol{F}_{hq}$
is an $f_{s}T_{q}$-dimensional square matrix such that: for any frequency
$f$ and all $1\leq k,i\leq f_{s}T_{q}$,
\begin{align}
\left(\boldsymbol{F}_{hq}\right)_{k,i}\left(f\right) & :=\sum_{l=-\infty}^{\infty}\mathcal{R}_{hq}\left(t_{k},t_{i+lf_{s}T_{q}}\right)\exp\left(j2\pi lf\right)\label{eq:DefnFhq}
\end{align}
and
\begin{align}
\left(\boldsymbol{F}_{qq}\right)_{k,i} & \left(f\right):=\sum_{l=-\infty}^{\infty}\mathcal{R}_{q}\left(t_{k},t_{i+lf_{s}T_{q}}\right)\exp\left(j2\pi lf\right).
\end{align}
Besides, for every frequency $f$, we define an $f_{q}T_{s}\times\infty$
dimensional matrix $\boldsymbol{F}_{q}\left(f\right)$ and an infinite
square diagonal matrix $\boldsymbol{F}_{h}\left(f\right)$ such that
for all $l\in\mathbb{Z}$ and $1\leq k\leq f_{q}T_{s}$:
\begin{align}
\left(\boldsymbol{F}_{q}\right)_{k,l}\left(f\right) & :=Q_{k}\left(f+lf_{q}\right),\\
\left(\boldsymbol{F}_{h}\right)_{l,l}\left(f\right) & :=H\left(f+lf_{q}\right),
\end{align}
where $Q_{k}(f)\overset{\Delta}{=}\mathcal{F}\left(q_{k}(\cdot)\right)=\mathcal{F}\left(q(t_{k},\cdot)\right)$.

The key properties of the above autocorrelation functions and Fourier
series are summarized in the following lemma.

\begin{lem}\label{lemma-Fhq-Fqq}The Fourier series matrices satisfy:
\begin{equation}
\boldsymbol{F}_{hq}=\boldsymbol{F}_{q}\boldsymbol{F}_{h}\boldsymbol{F}_{h}^{*}\boldsymbol{F}_{q}^{*}
\end{equation}
and
\begin{equation}
\boldsymbol{F}_{qq}=\boldsymbol{F}_{q}\boldsymbol{F}_{q}^{*}.
\end{equation}
\end{lem}

\begin{IEEEproof}See Appendix \ref{sec:Proof-of-Lemma-Fhq-Fqq}.\end{IEEEproof}

\subsection{Noise whitening}

Denote by $\mathcal{Q}_{k}\left(\cdot\right)$ the sampling operator
associated with the sample time $t_{k}$ such that $\mathcal{Q}_{k}\left(x\right)\overset{\Delta}{=}\int_{-\infty}^{\infty}q\left(t_{k},\tau\right)x\left(\tau\right)\mathrm{d}\tau$.
The correlation of noise components $\mathcal{Q}_{k}\left(\eta\right)$
at different times can be calculated as
\begin{align*}
 & \mathbb{E}\left[\mathcal{Q}_{k}\left(\eta\right)\mathcal{Q}_{l}^{*}\left(\eta\right)\right]\\
 & \text{ }\text{ }=\mathbb{E}\left[\int_{-\infty}^{\infty}q\left(t_{k},\tau_{k}\right)\eta\left(\tau_{k}\right)\mathrm{d}\tau_{k}\left(\int_{-\infty}^{\infty}q\left(t_{l},\tau_{l}\right)\eta\left(\tau_{l}\right)\mathrm{d}\tau_{l}\right)^{*}\right]\\
 & \text{ }\text{ }=\int_{-\infty}^{\infty}\int_{-\infty}^{\infty}q\left(t_{k},\tau_{k}\right)q^{*}\left(t_{l},\tau_{l}\right)\mathbb{E}\left(\eta\left(\tau_{k}\right)\eta^{*}\left(\tau_{l}\right)\right)\mathrm{d}\tau_{k}\mathrm{d}\tau_{l}\\
 & \text{ }\text{ }=\int_{-\infty}^{\infty}q\left(t_{k},\tau\right)q^{*}\left(t_{l},\tau\right)\mathrm{d}\tau,
\end{align*}
which immediately implies that $\mathcal{Q}\left(\eta\right)=\left[\cdots,\mathcal{Q}_{1}\left(\eta\right),\mathcal{Q}_{2}\left(\eta\right),\cdots\right]^{\top}$
is a zero-mean Gaussian vector with covariance matrix $\mathcal{R}_{q}$.

We now constrain both the transmit interval and the observation interval
to $\left[-nT_{q},nT_{q}\right]$. Let
\[
\boldsymbol{y}_{n}=\left[y\left[-nf_{s}T_{q}+1\right],,\cdots,y\left[nf_{s}T_{q}-1\right]y\left[nf_{s}T_{q}\right]\right]^{\top},
\]
where the sampled output sequence satisfies
\begin{equation}
y[k]=\mathcal{Q}_{k}\left(h(t)*x(t)\right)+\mathcal{Q}_{k}\left(\eta(t)\right).
\end{equation}
Introduce two $2nf_{s}T_{q}$-dimensional \emph{truncated} autocorrelation
matrices $\mathcal{R}_{hq}^{n}$ and $\mathcal{R}_{q}^{n}$ such that
for all $-nf_{s}T_{q}<k,l\leq nf_{s}T_{q}$,
\begin{align*}
\left(\mathcal{R}_{hq}^{n}\right)_{k,l} & =\mathcal{R}_{hq}\left(t_{k},t_{l}\right),\\
\left(\mathcal{R}_{q}^{n}\right)_{k,l} & =\mathcal{R}_{q}\left(t_{k},t_{l}\right).
\end{align*}
Clearly, the noise components of $\boldsymbol{y}_{n}$ exhibit a covariance
matrix $\mathcal{R}_{q}^{n}$, which motivates to whiten it first. 

By left multiplying $\boldsymbol{y}_{n}$ with $\left(\mathcal{R}_{q}^{n}\right)^{-\frac{1}{2}}$,
we obtain a new input-output relation as
\[
\tilde{y}_{n}[k]=\tilde{\mathcal{Q}}_{k}\left(h(t)*x\left(t\right)\right)+\tilde{\eta}\left[k\right],\quad\forall k\text{}\left(\left|k\right|\leq nf_{s}T_{q}\right),
\]
where $\left\{ \tilde{\eta}\left[k\right]\right\} $ are i.i.d. Gaussian
random variables each of unit variance. Denote by $\tilde{q}\left(t_{k},\tau\right)$
the equivalent impulse response of this new system. The truncated
output autocorrelation function $\mathcal{R}_{\tilde{q}}^{n}$ is
given as $\left(\mathcal{R}_{\tilde{q}}^{n}\right)_{k,l}=\mathcal{R}_{\tilde{q}}\left(t_{k},t_{l}\right)=\int_{-\infty}^{\infty}\tilde{q}\left(t_{k},\tau\right)\tilde{q}^{*}\left(t_{l},\tau\right)\mathrm{d}\tau$,
satisfying
\begin{equation}
\mathcal{R}_{\tilde{q}}^{n}=\left(\mathcal{R}_{q}^{n}\right)^{-\frac{1}{2}}\mathcal{R}_{hq}^{n}\left(\mathcal{R}_{q}^{n}\right)^{-\frac{1}{2}}
\end{equation}
by construction.

\subsection{Capacity via asymptotic properties of block Toeplitz matrices}

While both $\mathcal{R}_{q}^{n}$ and $\mathcal{R}_{hq}^{n}$ are
block Toeplitz matrices, $\mathcal{R}_{\tilde{q}}^{n}$ is in general
not a block Toeplitz matrix. By exploiting the asymptotic equivalence
in Toeplitz matrix theory \cite{Gray06}, one can see that $\mathcal{R}_{\tilde{q}}^{n}$
is asymptotically equivalent to a block-Toeplitz matrix generated
by the Fourier series
\begin{align*}
 & \mathcal{F}\left(\mathcal{R}_{q}^{-\frac{1}{2}}\right)\mathcal{F}\left(\mathcal{R}_{hq}\right)\mathcal{F}\left(\mathcal{R}_{q}^{-\frac{1}{2}}\right)\\
 & \text{ }\text{ }=\left(\boldsymbol{F}_{q}\boldsymbol{F}_{q}^{*}\right)^{-\frac{1}{2}}\boldsymbol{F}_{q}\boldsymbol{F}_{h}\boldsymbol{F}_{h}^{*}\boldsymbol{F}_{q}^{*}\left(\boldsymbol{F}_{q}\boldsymbol{F}_{q}^{*}\right)^{-\frac{1}{2}}.
\end{align*}
Therefore, the asymptotic spectral properties of a block-Toeplitz
matrix (e.g. \cite{Tilli98}) state that for any nondecreasing continuous
function $g(t)$ with a bounded slope, one has
\begin{align}
\lim_{n\rightarrow\infty}\frac{1}{2nT_{q}}\sum_{i=1}^{2nf_{s}T_{q}}g\left(\lambda_{i}\left(\mathcal{R}_{\tilde{q}}^{n}\right)\right)= & \frac{1}{2\pi T_{q}}{\displaystyle \int}_{-\pi}^{\pi}\sum_{i=1}^{f_{s}T_{q}}g\left(\hat{\lambda}_{i}\right)\mathrm{d}\omega,\label{eq:AsymptoticSpectralProperty}
\end{align}
where $\hat{\lambda}_{i}$ represents the $i$th eigenvalue of $\left(\boldsymbol{F}_{q}\boldsymbol{F}_{q}^{*}\right)^{-\frac{1}{2}}\boldsymbol{F}_{q}\boldsymbol{F}_{h}\boldsymbol{F}_{h}^{*}\boldsymbol{F}_{q}^{*}\left(\boldsymbol{F}_{q}\boldsymbol{F}_{q}^{*}\right)^{-\frac{1}{2}}$.

(1) The capacity of the sampled channel with an optimal water level
$\nu_{\mathrm{p}}$ can now be calculated as
\begin{align}
C^{\mathcal{P}}\left(P\right)= & \lim_{n\rightarrow\infty}\frac{1}{2nT_{q}}\sum_{i=1}^{2nf_{s}T_{q}}\frac{1}{2}\left[\log\left(\nu_{\mathrm{p}}\lambda_{i}\left(\mathcal{R}_{\tilde{q}}^{n}\right)\right)\right]^{+}\label{eq:PerfectCSIPeriodicToeplitz}\\
= & \frac{1}{2\pi T_{q}}{\displaystyle \int}_{-\pi}^{\pi}\sum_{i=1}^{f_{s}T_{q}}\frac{1}{2}\left[\log\left(\nu_{\mathrm{p}}\hat{\lambda}_{i}\right)\right]^{+}\mathrm{d}\omega\label{eq:AsymptoticEquivalenceCp}\\
= & {\displaystyle \int}_{-f_{q}/2}^{f_{q}/2}\sum_{i=1}^{f_{s}T_{q}}\frac{1}{2}\left[\log\left(\nu_{\mathrm{p}}\hat{\lambda}_{i}\right)\right]^{+}\mathrm{d}f,\nonumber 
\end{align}
where (\ref{eq:AsymptoticEquivalenceCp}) is a consequence of (\ref{eq:AsymptoticSpectralProperty}).

The water level $\nu$ is computed through the following parametric
equation
\[
\lim_{n\rightarrow\infty}\frac{1}{2nT_{q}}\sum_{i=1}^{2nf_{s}T_{q}}\left[\nu_{\mathrm{p}}-\frac{1}{\lambda_{i}\left(\mathcal{R}_{\tilde{q}}^{n}\right)}\right]^{+}\mathrm{d}f=P,
\]
which by (\ref{eq:AsymptoticSpectralProperty}) is asymptotically
equivalent to
\[
\frac{1}{2\pi T_{q}}{\displaystyle \int}_{-\pi}^{\pi}\sum_{i=1}^{f_{s}T_{q}}\left[\nu_{\mathrm{p}}-\frac{1}{\hat{\lambda}_{i}}\right]^{+}\mathrm{d}\omega=P,
\]
or
\[
{\displaystyle \int}_{-f_{q}/2}^{f_{q}/2}\sum_{i=1}^{f_{s}T_{q}}\left[\nu_{\mathrm{p}}-\frac{1}{\hat{\lambda}_{i}}\right]^{+}\mathrm{d}f=P.
\]
by change of variables. This establishes the claim.

(2) We consider now the scenario where equal power allocation is employed.
Classical MIMO channel capacity results \cite{Tel1999} indicate that
the optimal power allocation for the transmitter is to allocate equal
amount of power in all transmit branches. It remains to see how much
power is allocated to the branch associated with $\lambda_{i}\left(\mathcal{R}_{\tilde{q}}^{n}\right)$. 

In fact, if the transmitter knows the channel bandwidth, almost all
power (except for negligible leakage due to finite-time approximation)
will be allocated inside the channel bandwidth $[0,W]$. Therefore,
by the Shannon-Nyquist sampling theorem, all transmit signals can
be equivalently transformed to a delta train $\sum_{i=-\infty}^{\infty}x_{i}\delta(t-i/W)$,
where $x_{i}$'s are randomly generated transmit signals. Consider
the input time block $\left[-nT_{q},nT_{q}\right]$, then there are
equivalently $2nT_{q}W$ transmit branches inside this time block.
Since the total power is $P_{\text{tot}}=2nT_{q}P$, the power allocated
to each transmit branch is given by
\[
P_{0}=\lim_{n\rightarrow\infty}\frac{P_{\text{tot}}}{2nT_{q}/\left(\frac{1}{W}\right)}=\lim_{n\rightarrow\infty}\frac{2nT_{q}P}{2nWT_{q}}=\frac{P}{W}.
\]
As a result, the sampled capacity under equal power allocation is
computed as
\begin{align*}
C_{\mathrm{eq}}^{\mathcal{P}}\left(P\right)= & \lim_{n\rightarrow\infty}\frac{1}{2nT_{q}}\sum_{i=1}^{2nf_{s}T_{q}}\frac{1}{2}\log\left(1+\frac{P}{W}\lambda_{i}\left(\mathcal{R}_{\tilde{q}}^{n}\right)\right)\\
= & \frac{1}{2\pi T_{q}}{\displaystyle \int}_{-\pi}^{\pi}\sum_{i=1}^{f_{s}T_{q}}\frac{1}{2}\log\left(1+\frac{P}{W}\hat{\lambda}_{i}\right)\mathrm{d}\omega\\
= & {\displaystyle \int}_{-f_{q}/2}^{f_{q}/2}\frac{1}{2}\log\left(\boldsymbol{I}+\frac{P}{W}\left(\boldsymbol{F}_{q}\boldsymbol{F}_{q}^{*}\right)^{-\frac{1}{2}}\boldsymbol{F}_{q}\boldsymbol{F}_{h}\right.\\
 & \quad\quad\quad\quad\left.\cdot\boldsymbol{F}_{h}^{*}\boldsymbol{F}_{q}^{*}\left(\boldsymbol{F}_{q}\boldsymbol{F}_{q}^{*}\right)^{-\frac{1}{2}}\right)\mathrm{d}f.
\end{align*}

\section{Proof of Lemma \ref{lem:PropertiesLambda}\label{sec:Proof-of-Lemma-PropertiesLambda} }

(1) Let $\mathcal{G}$ and $\tilde{\mathcal{G}}$ denote respectively
the operators associated with $g(t)$ and $\tilde{g}(t)$. Then, the
triangle inequality yields
\begin{align*}
\left\Vert \hat{\mathcal{P}}_{\text{w}}\mathcal{G}\right\Vert _{\text{F}} & \leq\left\Vert \hat{\mathcal{P}}_{\text{w}}\tilde{\mathcal{G}}\right\Vert _{\text{F}}+\left\Vert \hat{\mathcal{P}}_{\text{w}}\left(\mathcal{G}-\tilde{\mathcal{G}}\right)\right\Vert _{\text{F}}\\
 & \leq\left\Vert \hat{\mathcal{P}}_{\text{w}}\tilde{\mathcal{G}}\right\Vert _{\text{F}}+\left\Vert \mathcal{G}-\tilde{\mathcal{G}}\right\Vert _{\text{F}},
\end{align*}
and hence
\begin{align}
\left\Vert \hat{\mathcal{P}}_{\text{w}}\mathcal{G}\right\Vert _{\text{F}}^{2} & \leq\left\Vert \hat{\mathcal{P}}_{\text{w}}\tilde{\mathcal{G}}\right\Vert _{\text{F}}^{2}+\left\Vert \mathcal{G}-\tilde{\mathcal{G}}\right\Vert _{\text{F}}^{2}+2\left\Vert \tilde{\mathcal{G}}\right\Vert _{\text{F}}\left\Vert \mathcal{G}-\tilde{\mathcal{G}}\right\Vert _{\text{F}}.\label{eq:FrobeniusNormDifference}
\end{align}
From (\ref{eq:ChoiceL1TruncatedChannel}) one can easily show that
for any $\xi>0$, there exists a $T_{0}$ such that for every $T>T_{0}$,
one has 
\[
\left\Vert \mathcal{G}-\tilde{\mathcal{G}}\right\Vert _{\text{F}}\leq\sqrt{2T\left(\int_{-\infty}^{-T}+\int_{T}^{\infty}\right)\left|h(t)\right|^{2}\mathrm{d}t}\leq\sqrt{2T\xi}.
\]
Additionally, suppose that $\int_{-\infty}^{\infty}\left|h(t)\right|^{2}\mathrm{d}t\leq C_{g}<\infty$.
Then, we have
\[
\left\Vert \tilde{\mathcal{G}}\right\Vert _{\text{F}}\leq\sqrt{2T\int_{-\infty}^{\infty}\left|h(t)\right|^{2}\mathrm{d}t}\leq\sqrt{2TC_{g}}.
\]
This together with (\ref{eq:FrobeniusNormDifference}) immediately
gives us
\[
\left\Vert \hat{\mathcal{P}}_{\text{w}}\mathcal{G}\right\Vert _{\text{F}}^{2}\leq\left\Vert \hat{\mathcal{P}}_{\text{w}}\tilde{\mathcal{G}}\right\Vert _{\text{F}}^{2}+2T\xi+4T\sqrt{\xi C_{g}}.
\]

Similar to \cite[Theorem 8.4.1]{Gallager68}, we can obtain that
\[
\sum_{i}\lambda_{i}=\left\Vert \hat{\mathcal{P}}_{\text{w}}\mathcal{G}\right\Vert _{\text{F}}^{2}\quad\text{and}\quad\sum_{i}\tilde{\lambda}_{i}=\left\Vert \hat{\mathcal{P}}_{\text{w}}\tilde{\mathcal{G}}\right\Vert _{\text{F}}^{2}.
\]
Therefore,
\begin{align*}
\frac{1}{2T}\sum_{i}\lambda_{i}-\frac{1}{2T}\sum_{i}\tilde{\lambda}_{i} & =\frac{1}{2T}\left\Vert \hat{\mathcal{P}}_{\text{w}}\mathcal{G}\right\Vert _{\text{F}}^{2}-\frac{1}{2T}\left\Vert \hat{\mathcal{P}}_{\text{w}}\tilde{\mathcal{G}}\right\Vert _{\text{F}}^{2}\\
 & \leq\xi+2\sqrt{\xi C_{g}}.
\end{align*}
Similarly,
\begin{align*}
\frac{1}{2T}\sum_{i}\lambda_{i}-\frac{1}{2T}\sum_{i}\tilde{\lambda}_{i} & \geq-\xi-2\sqrt{\xi C_{g}}.
\end{align*}

(2) We can also bound the sum of eigenvalues as follows 
\begin{align*}
\frac{1}{2T}\sum_{i}\lambda_{i} & =\frac{1}{2T}\left\Vert \hat{\mathcal{P}}_{\text{w}}\mathcal{G}\right\Vert _{\text{F}}^{2}\leq\frac{1}{2T}\left\Vert \mathcal{G}\right\Vert _{\text{F}}^{2}\\
 & =\int_{-\infty}^{\infty}\left|h(t)\right|^{2}\mathrm{d}t<\infty.
\end{align*}

(3) If $g(t)=O\left(\frac{1}{t^{1+\epsilon}}\right)$, then one can
further control
\begin{align*}
\left\Vert \mathcal{G}-\tilde{\mathcal{G}}\right\Vert _{\text{F}}^{2} & \leq2T\left(\int_{-\infty}^{-T}+\int_{T}^{\infty}\right)\left|h(t)\right|^{2}\mathrm{d}t\\
 & \leq2TO\left(\frac{1}{T^{2+2\epsilon}}\right)=O\left(\frac{1}{T^{1+2\epsilon}}\right).
\end{align*}
Therefore, applying Weyl's Theorem \cite{ipsen2009refined} yields
that
\begin{align*}
\left|\lambda_{i}-\tilde{\lambda}_{i}\right| & \leq\left\Vert \hat{\mathcal{P}}_{\text{w}}\mathcal{G}\left(\hat{\mathcal{P}}_{\text{w}}\mathcal{G}\right)^{*}-\hat{\mathcal{P}}_{\text{w}}\tilde{\mathcal{G}}\left(\hat{\mathcal{P}}_{\text{w}}\tilde{\mathcal{G}}\right)^{*}\right\Vert _{\text{F}}\\
 & \leq\left\Vert \hat{\mathcal{P}}_{\text{w}}\left(\mathcal{G}-\tilde{\mathcal{G}}\right)\right\Vert _{\text{F}}\left(\left\Vert \hat{\mathcal{P}}_{\text{w}}\mathcal{G}\right\Vert _{\text{F}}+\left\Vert \hat{\mathcal{P}}_{\text{w}}\tilde{\mathcal{G}}\right\Vert _{\text{F}}\right)\\
 & \leq\left\Vert \mathcal{G}-\tilde{\mathcal{G}}\right\Vert _{\text{F}}\left(\left\Vert \mathcal{G}\right\Vert _{\text{F}}+\left\Vert \tilde{\mathcal{G}}\right\Vert _{\text{F}}\right)\\
 & \leq O\left(\frac{1}{T^{\epsilon}}\right).
\end{align*}
Therefore, for any small $\xi>0$, there exists a constant $T_{0,\epsilon}$
such that for every $T>T_{0,\epsilon}$, one has $\left|\lambda_{i}-\tilde{\lambda}_{i}\right|<\xi.$

\section{Proof of Lemma \ref{lemma-Fhq-Fqq}\label{sec:Proof-of-Lemma-Fhq-Fqq}}

Simple manipulation yields 
\begin{align*}
 & \mathcal{R}_{hq}\left(t_{k},t_{l}\right)=\int_{-\infty}^{\infty}s\left(t_{k},\tau\right)s^{*}\left(t_{l},\tau\right)\mathrm{d}\tau\\
 & \text{ }\text{ }=\iiint q\left(t_{k},\tau_{k}\right)h\left(\tau_{k}-\tau\right)h^{*}\left(\tau_{l}-\tau\right)q^{*}\left(t_{l},\tau_{l}\right)\mathrm{d}\tau_{k}\mathrm{d}\tau_{l}\mathrm{d}\tau\\
 & \text{ }\text{ }=\iint q_{k}\left(\tau_{k}\right)\mathcal{R}_{h}\left(\tau_{l}-\tau{}_{k}\right)q_{l}^{*}\left(\tau_{l}\right)\mathrm{d}\tau_{k}\mathrm{d}\tau_{l},
\end{align*}
where
\begin{align*}
\mathcal{R}_{h}\left(\tau_{l}-\tau{}_{k}\right) & :=\int_{\tau}h\left(\tau_{k}-\tau\right)h^{*}\left(\tau_{l}-\tau\right)\mathrm{d}\tau\\
 & =\int_{\tau}h\left(\tau_{k}-\tau_{l}+\tau\right)h^{*}\left(\tau\right)\mathrm{d}\tau\\
 & =\left(h*h^{-*}\right)\left(\tau_{k}-\tau_{l}\right).
\end{align*}
Here, for any function $f(t)$, we use $f^{-}(t)$ to denote $f(-t)$. 

By the periodicity assumption of the sampling system, one can derive
\begin{align*}
 & \mathcal{R}_{hq}\left(t_{k+af_{s}T_{q}},t_{l+bf_{s}T_{q}}\right)\\
 & \text{ }\text{ }=\iint q\left(t_{k}+aT_{q},\tau_{k}\right)\mathcal{R}_{h}\left(\tau_{l}-\tau{}_{k}\right)q^{*}\left(t_{l}+bT_{q},\tau_{l}\right)\mathrm{d}\tau_{k}\mathrm{d}\tau_{l}\\
 & \text{ }\text{ }=\iint q\left(t_{k},\tau_{k}-aT_{q}\right)\mathcal{R}_{h}\left(\tau_{l}-\tau{}_{k}\right)\\
 & \quad\quad\quad\quad\quad\quad\cdot q^{*}\left(t_{l}+\left(b-a\right)T_{q},\tau_{l}-aT_{q}\right)\mathrm{d}\tau_{k}\mathrm{d}\tau_{l}\\
 & \text{ }\text{ }=\mathcal{R}_{hq}\left(t_{k},t_{l+\left(b-a\right)f_{s}T_{q}}\right).
\end{align*}
Observing that
\begin{align*}
 & \mathcal{R}_{hq}\left(t_{k},t_{i+lf_{s}T_{q}}\right)\\
 & \text{ }\text{ }=\iint q\left(t_{k},\tau_{k}\right)\mathcal{R}_{h}\left(\tau_{i}-\tau_{k}\right)q^{*}\left(t_{i}+lT_{q},\tau_{i}\right)\mathrm{d}\tau_{k}\mathrm{d}\tau_{i}\\
 & \text{ }\text{ }=\iint q_{k}\left(\tau_{k}\right)\mathcal{R}_{h}\left(\tau_{i}+lT_{q}-\tau_{k}\right)q_{i}^{*}\left(\tau_{i}\right)\mathrm{d}\tau_{k}\mathrm{d}\tau_{i}\\
 & \text{ }\text{ }=\left(\mathcal{R}_{h}*q_{k}*q_{i}^{-*}\right)\left(lT_{q}\right),
\end{align*}
we can see that $\left(\boldsymbol{F}_{hq}\right)_{k,i}$ is simply
the Fourier transform of the sampled sequence of $\mathcal{R}_{h}*q_{k}*q_{i}^{-*}$.
The properties of the Fourier transform suggest that
\begin{align*}
\mathcal{F}\left(\mathcal{R}_{h}*q_{k}*q_{i}^{-*}\right)\left(f\right) & =\mathcal{F}\left(\mathcal{R}_{h}\right)\left(f\right)\cdot Q_{k}\left(f\right)\cdot Q_{i}^{*}\left(f\right)\\
 & =\left|H\left(f\right)\right|^{2}Q_{k}\left(f\right)\cdot Q_{i}^{*}\left(f\right),
\end{align*}
where $Q_{k}(f):=\mathcal{F}\left(q_{k}\right)$. By definition in
(\ref{eq:DefnFhq}), one can write
\begin{align*}
\left(\boldsymbol{F}_{hq}\right)_{k,i}\left(f\right) & :=\sum_{l=-\infty}^{\infty}\mathcal{R}_{hq}\left(t_{k},t_{i+lf_{s}T_{q}}\right)\exp\left(j2\pi lf\right)\\
 & =\sum_{l=-\infty}^{\infty}\left(\mathcal{R}_{h}*q_{k}*q_{i}^{-*}\right)\left(lT_{q}\right)\exp\left(j2\pi lf\right),
\end{align*}
which immediately leads to
\[
\left(\boldsymbol{F}_{hq}\right)_{k,i}=\sum_{l=-\infty}^{\infty}Q_{k}\left(f+lf_{q}\right)\left|H\left(f+lf_{q}\right)\right|^{2}Q_{i}^{*}\left(f+lf_{q}\right).
\]
This allows us to express $\boldsymbol{F}_{hq}$ as
\begin{equation}
\boldsymbol{F}_{hq}=\boldsymbol{F}_{q}\boldsymbol{F}_{h}\boldsymbol{F}_{h}^{*}\boldsymbol{F}_{q}^{*}.\label{eq:Fhq}
\end{equation}

Similarly, the equality $\boldsymbol{F}_{qq}=\boldsymbol{F}_{q}\boldsymbol{F}_{q}^{*}$
is then an immediate consequence of (\ref{eq:Fhq}) by setting $\boldsymbol{F}_{h}=\boldsymbol{I}$.

\bibliographystyle{IEEEtran} \bibliographystyle{IEEEtran} \bibliographystyle{IEEEtran}
\bibliography{bibfileNonuniform}

\begin{IEEEbiographynophoto}{Yuxin Chen} (S'09) received the B.S. in Microelectronics with High Distinction from Tsinghua University in 2008, the M.S. in Electrical and Computer Engineering from the University of Texas at Austin in 2010, and the M.S. in Statistics from Stanford University in 2013. He is currently a Ph.D. candidate in the Department of Electrical Engineering at Stanford University. His research interests include information theory, compressed sensing, network science and high-dimensional statistics. \end{IEEEbiographynophoto} \begin{IEEEbiographynophoto}{Yonina C. Eldar} (S'98-M'02-SM'07-F'12) received the B.Sc. degree in physics and  the B.Sc. degree in electrical engineering both from Tel-Aviv University (TAU), Tel-Aviv,  Israel, in 1995 and 1996, respectively, and the Ph.D. degree in electrical engineering and computer science from the  Massachusetts Institute of Technology (MIT), Cambridge, in 2002.

 From January 2002 to July 2002, she was a Postdoctoral Fellow at the  Digital Signal Processing Group at MIT. She is currently a Professor in the Department of Electrical  Engineering at the Technion-Israel Institute of Technology, Haifa and holds the  The Edwards Chair in Engineering. She is  also a Research Affiliate with the Research Laboratory of Electronics at MIT  and a Visiting Professor at Stanford University, Stanford, CA. Her research interests are in the broad areas of  statistical signal processing, sampling theory and compressed sensing,  optimization methods, and their applications to biology and optics.

Dr. Eldar was in the program for outstanding students at TAU from 1992 to 1996. In 1998, she held the Rosenblith Fellowship for study in electrical engineering at MIT, and in 2000, she held an IBM Research Fellowship. From  2002 to 2005, she was a Horev Fellow of the Leaders in Science and  Technology program at the Technion and an Alon Fellow. In 2004, she was  awarded the Wolf Foundation Krill Prize for Excellence in Scientific  Research, in 2005 the Andre and Bella Meyer Lectureship, in 2007 the Henry  Taub Prize for Excellence in Research, in 2008 the Hershel Rich Innovation  Award, the Award for Women with Distinguished Contributions, the Muriel \& David Jacknow Award for Excellence in Teaching, and the Technion Outstanding  Lecture Award, in 2009 the Technion's Award for Excellence in Teaching, in 2010 the Michael  Bruno Memorial Award from the Rothschild Foundation, and in 2011 the Weizmann Prize for Exact Sciences.   In 2012 she was elected to the Young Israel Academy of Science and to the Israel Committee for Higher Education, and elected an IEEE Fellow.  In 2013 she received the Technion's Award for Excellence in Teaching, the Hershel Rich Innovation Award, and the IEEE Signal Processing Technical Achievement Award. She received several best paper awards together with her research students and colleagues.  She received several best paper awards together with her research students and colleagues. She is the Editor in Chief of Foundations and Trends in Signal Processing. In the past, she was a Signal Processing Society Distinguished Lecturer, member of the IEEE Signal Processing Theory and Methods and Bio Imaging Signal Processing technical committees, and served as an associate editor for the IEEE Transactions On Signal Processing, the EURASIP Journal of Signal Processing, the SIAM Journal on Matrix Analysis and Applications, and the SIAM Journal on Imaging Sciences.
\end{IEEEbiographynophoto} 

\begin{IEEEbiographynophoto}{Andrea J. Goldsmith} (S'90-M'93-SM'99-F'05)
is the Stephen Harris professor in the School of Engineering and a professor of Electrical Engineering at Stanford University. She was previously on the faculty of Electrical Engineering at Caltech. Her research interests are in information theory and communication theory, and their application to wireless communications and related fields. She co-founded and serves as Chief Scientist of Accelera, Inc., and previously co-founded and served as CTO of Quantenna Communications, Inc. She has also held industry positions at Maxim Technologies, Memorylink Corporation, and AT\&T Bell Laboratories. Dr. Goldsmith is a Fellow of the IEEE and of Stanford, and she has received several awards for her work, including the IEEE Communications Society and Information Theory Society joint paper award, the IEEE Communications Society Best Tutorial Paper Award, the National Academy of Engineering Gilbreth Lecture Award, the IEEE ComSoc Communications Theory Technical Achievement Award, the IEEE ComSoc Wireless Communications Technical Achievement Award, the Alfred P. Sloan Fellowship, and the Silicon Valley/San Jose Business Journal's Women of Influence Award. She is author of the book ``Wireless Communications'' and co-author of the books ``MIMO Wireless Communications'' and ``Principles of Cognitive Radio,'' all published by Cambridge University Press. She received the B.S., M.S. and Ph.D. degrees in Electrical Engineering from U.C. Berkeley.  

Dr. Goldsmith has served on the Steering Committee for the IEEE Transactions on Wireless Communications and as editor for the IEEE Transactions on Information Theory, the Journal on Foundations and Trends in Communications and Information Theory and in Networks, the IEEE Transactions on Communications, and the IEEE Wireless Communications Magazine. She participates actively in committees and conference organization for the IEEE Information Theory and Communications Societies and has served on the Board of Governors for both societies. She has also been a Distinguished Lecturer for both societies, served as President of the IEEE Information Theory Society in 2009, founded and chaired the student committee of the IEEE Information Theory society, and chaired the Emerging Technology Committee of the IEEE Communications Society. At Stanford she received the inaugural University Postdoc Mentoring Award, served as Chair of Stanfords Faculty Senate in 2009 and currently serves on its Faculty Senate and on its Budget Group.
\end{IEEEbiographynophoto}

\end{document}